\documentclass[twoside,runningheads]{llncs}
\usepackage{cite}
\usepackage{graphicx}
\usepackage{amsmath}
\newtheorem{notation}{Notation}
\usepackage{hyperref}

\usepackage[createShortEnv]{proof-at-the-end}
\newcommand{\appendixLocation}{appendix (Def. \ref{semantics})}
\newcommand{\appendixLocationTwo}{Lemma \ref{test_conversion} in the appendix}
\newcommand{\fullproofLocation}{See \hyperref[proof:prAtEnd\pratendcountercurrent]{full proof} on page~\pageref{proof:prAtEnd\pratendcountercurrent}.}

\usepackage{xcolor}
\usepackage{tabularx}
\usepackage{verbatimbox}
\usepackage{xspace}

\usepackage{prettyref}
\usepackage{wrapfig}
\usepackage{cutwin}
\usepackage{mathtools}
\usepackage{adjustbox}
\usepackage{float}
\usepackage{longtable}
\usepackage{stackengine}
\newcommand{\rref}[2][]{\prettyref{#2}}
\newrefformat{sec}{Section\,\ref{#1}}
\newrefformat{def}{Def.\,\ref{#1}}
\newrefformat{thm}{Theorem\,\ref{#1}}
\newrefformat{prop}{Proposition\,\ref{#1}}
\newrefformat{lem}{Lemma\,\ref{#1}}
\newrefformat{cor}{Corollary\,\ref{#1}}
\newrefformat{ex}{Example\,\ref{#1}}
\newrefformat{tab}{Table\,\ref{#1}}
\newrefformat{fig}{Fig.\,\ref{#1}}
\newrefformat{app}{Appendix\,\ref{#1}}
\usepackage{stmaryrd}
\usepackage{amssymb}
\newcommand{\dGLsc}{\text{\upshape\textsf{d{\kern-0.05em}G{\kern-0.15em}$\mathcal{L}_{sc}$}}\xspace}
\usepackage{bussproofs}
\usepackage{xparse}

\newcommand{\dDL}[1][]
{\text{\upshape\textsf{d{\kern-0.1em}$\mathcal{L}$}}\xspace}
\newcommand{\dGL}[1][]
{\text{\upshape\textsf{d{\kern-0.05em}G{\kern-0.15em}$\mathcal{L}$}}\xspace}
\newcommand{\angeldGLA}[3]{\langle #1 \rangle(#2\textcolor{blue}{\mathbf{,}}\, #3)}
\newcommand{\demondGLA}[3]{[#1](#2\textcolor{blue}{\mathbf{,}}\, #3)}
\newcommand{\angeldGL}[2]{\langle #1 \rangle #2}
\newcommand{\demondGL}[2]{[#1]#2}
\NewDocumentCommand\angel{mmg}{
    \IfNoValueTF{#3}{\angeldGL{#1}{#2}}{\angeldGLA{#1}{#2}{#3}}
}
\NewDocumentCommand\demon{mmg}{
    \IfNoValueTF{#3}{\demondGL{#1}{#2}}{\demondGLA{#1}{#2}{#3}}
}
\newcommand{\winr}[1]{\llbracket #1 \rrbracket}
\newcommand{\winrAdglA}[3]{\varsigma_{#1}(#2, #3)}
\newcommand{\winrAdgl}[2]{\varsigma_{#1}(#2)}
\newcommand{\winrDdglA}[3]{\delta_{#1}(#2, #3)}
\newcommand{\winrDdgl}[2]{\delta_{#1}(#2)}
\NewDocumentCommand\winrA{mmg}{
    \IfNoValueTF{#3}{\winrAdgl{#1}{#2}}{\winrAdglA{#1}{#2}{#3}}
}
\NewDocumentCommand\winrD{mmg}{
    \IfNoValueTF{#3}{\winrDdgl{#1}{#2}}{\winrDdglA{#1}{#2}{#3}}
}

\newcommand{\C}{\mathbf{C}}
\newcommand{\s}{\mathcal{S}}

\newcommand{\asA}{\langle := \rangle}
\newcommand{\asD}{[:=]}
\newcommand{\conA}{\langle ' \rangle}
\newcommand{\conD}{[']}
\newcommand{\tA}{\langle ? \rangle}
\newcommand{\tD}{[?]}
\newcommand{\chA}{\langle \cup \rangle}
\newcommand{\chD}{[\cup]}
\newcommand{\seqA}{\langle {;} \rangle}
\newcommand{\seqD}{[{;}]}
\newcommand{\dA}{\langle^d\rangle}
\newcommand{\dD}{[^d]}
\newcommand{\rA}{\langle^*\rangle}
\newcommand{\rD}{[^*]}
\newcommand{\FP}{\text{FP}}
\newcommand{\ind}{\text{ind}}
\newcommand{\AD}{\langle \land ]}
\newcommand{\impAD}{\langle\land]_2}
\newcommand{\MA}{\text{M}\langle\rangle}
\newcommand{\MAtwo}{\text{M}_2\langle\rangle}
\newcommand{\detR}{\text{det}}
\newcommand{\spA}{\text{S}\langle\rangle}
\newcommand{\re}{\text{re}\langle\rangle}
\newcommand{\reD}{\text{re}[]}
\newcommand{\id}{\text{id}^{-d}}
\newcommand{\FPtwo}{\text{FP}_2}

\newcommand{\cuswidth}{1.1cm}

\usepackage{ifpdf}
\ifpdf
\fi

\begin{document}
\title{Semi-Competitive Differential Game Logic}
\author{Julia Butte\orcidID{0009-0003-5066-8412} \and Andr\'e Platzer\orcidID{0000-0001-7238-5710}}
\institute{
  Karlsruhe Institute of Technology, Karlsruhe, Germany
  \email{$\{$julia.butte$\|$platzer$\}$@kit.edu}
}

\maketitle
\begin{abstract}
    This paper introduces \emph{semi-competitive differential game lo\-gic} \dGLsc, which enables verification of safety-critical applications that involve interactions between two agents. In \dGLsc, these interactions are specified as games on hybrid systems with two players that \emph{may} collaborate with each other when helpful and \emph{may} compete when necessary.
    The players in the hybrid games of \dGLsc have individual goals that may overlap, leading to nonzero-sum games. 
    This makes \dGLsc especially well-suited for verifying situations where players, e.g., share safety objectives but otherwise pursue different goals, so that zero-sum assumptions lead to overly conservative results.
    Additionally, \dGLsc solves the subtlety that even though each player may benefit from knowledge of the other player's goals, e.g., concerning shared safety objectives, unsafe situations might still occur if every player were to mutually assume the other player would act to avoid unsafety.
    The syntax and semantics, as well as a sound and relatively complete proof calculus are presented for \dGLsc.
    The relationship between \dGLsc and zero-sum differential game logic \dGL is discussed and the purpose of \dGLsc illustrated in a canonical example.
\end{abstract}

\section{Introduction}

The safety of cyber-physical systems (CPS) is of significant interest to avoid damage to persons and goods due to faulty programs \cite{Platzer2018,DBLP:journals/tcad/KabraMP22,DBLP:journals/tecs/CleavelandMP23}.
Cyber-physical systems include trains, planes, robots and autonomous cars \cite{DBLP:journals/tcad/KabraMP22,Lam2017}.
Particularly challenging are situations involving two CPSs, due to their possible interactions which frequently occur in aerial collision avoidance or steering of autonomous cars.
To verify, for example, two planes on a collision course, different approaches can be used: The situation could be regarded as one hybrid system, and its safety verified using differential dynamic logic (\dDL) \cite{DBLP:journals/jar/Platzer08,Platzer10,DBLP:journals/jar/Platzer17,Platzer2018,DBLP:journals/jacm/PlatzerT20,DBLP:conf/tacas/JeanninGKGSZP15}. 
This corresponds to all planes being centrally controlled, which is infeasible for larger numbers of planes.
Or the situation can be modeled with games, regarding each plane as a player which represents the fact that planes are normally steered independently. 
Collision avoidance could then be verified assuming adversarial players using differential game logic (\dGL) \cite{DBLP:journals/tecs/CleavelandMP23, Platzer15}.
This assumption results in a zero-sum game, i.e. exactly one player wins.
But if one pilot is trying to avoid a collision, the intruder pilot does \emph{not} in fact have the opposite goal of causing a collision!
Flying under a zero-sum assumption gives safe but unnecessarily conservative results.

Instead, many games are non-zero-sum, where players have individual goals which often overlap.
The planes, for example, both do not want to crash, but they still want to fly into different directions.
Such partially shared goals open up the possibility of cooperating, which leads to new winning strategies.
Alas, the possibility and will to cooperate is no guarantee that everything goes smoothly.
Misunderstandings might still happen, leading for example, to unintentional crashes.

To verify the safety of these two-player non-zero sum games, this paper introduces the \emph{semi-competitive differential game logic} (\dGLsc) which offers a natural way of expressing and reasoning about non-zero sum hybrid games.
\dGLsc supports players that behave \emph{semi-competitively}, i.e., both players have individual goals and are open for cooperation but will compete if necessary.
They both try to reach their goal while being aware of the other player's goal and can collaborate to jointly meet their respective goals.
This strategy provides a solution to the challenging question of when players should cooperate. The possibility of cooperation adds more winning strategies thus enabling safety proofs for situations that are unsafe under conservative zero-sum game assumptions. 

To underline the semantics' suitability, the paper derives an alternative representation of semi-competitive games in \dGL. Theoretically, this is possible, since all mixed (co)inductive concepts are definable in \dGL \cite{Platzer15}. But as \dGL is fundamentally designed for zero-sum, this requires manual coding, causes redundancies, and duplicates verification effort.

The structure of this paper is as follows: First, Section \ref{related_work} relates and compares relevant logics to \dGLsc. After that, the logic \dGL whose correspondence to \dGLsc  will be proved later on, is reviewed in Section \ref{preliminaries}. Then, Section \ref{main} introduces the notion of semi-competitiveness and defines the syntax and semantics of \dGLsc. Section \ref{properties} establishes important properties of \dGLsc and its relation to \dGL. Afterward, a proof calculus for \dGLsc is introduced in Section \ref{proof_calculus} and its soundness and relative completeness proven. 

\section{Related Work}
\label{related_work}

Non-zero sum games is a wide field of study that has been addressed by various communities. Game theory provides fundamental definitions of non-zero sum games and studies of their equilibria \cite{vonNeumannMorgenstern+2004, Nash50, nash2024} which lay the basis for our work. In terms of game theory, the games played in \dGLsc are non-zero sum sequential games with perfect information and binary payoffs. Whether a player can reach their goal can be considered by backward induction in \dGLsc, similar to the determination of a subgame perfect equilibrium \cite{Bielefeld1988}. But \dGLsc has the advantage that the backward induction is done without game trees, which might have uncountable, infinite breadth and unbounded depth due to continuous dynamics and consequently would be hard to handle.

In the field of synthesis, multi-player non-zero sum games played on a graph have been investigated by Fisman et al. who developed rational synthesis \cite{Fisman2010}. Although the assumptions about the players' behavior are very similar to \dGLsc, their work pursues a different aim: Rational synthesis can be used to construct correct systems, while \dGLsc instead can be used to verify already existing systems. 

In purely discrete settings, non-zero sum games have already been addressed by Chaterjee et al. who introduced strategy logic \cite{CHATTERJEE2010677}. This is a logic that allows modelling two-player non-zero sum games played on a graph. As explicit quantifications over strategies are possible, the logic is more powerful than e.g. alternating-time temporal logic (ATL) \cite{Alur2002} or propositional game logic \cite{parikh85} but also more complex. Therefore, it is better suited for theoretical investigations than for practical use, which \dGLsc supports. Additionally, continuous dynamics cannot be expressed in strategy logic. 

For a hybrid setting, there exist rectangular hybrid games \cite{Henzinger19999} or STORMED games \cite{VLADIMEROU20116770} which both model players with hybrid automata. While STORMED games are able to express more varieties of continuous dynamics than rectangular hybrid games, it is possible for both of them to perform model checking and controller synthesis. But unlike \dGLsc, both of these hybrid games assume the players to act adversarially. Another work on hybrid games which also assumes adversarial players, by Mitchell et al. \cite{Mitchell2005}, computes backwards reachable sets for hybrid games which resemble winning regions in \dGLsc. Unlike \dGLsc, they do not use a comparatively simple state-based semantics, but instead compute these sets by numerically solving Hamilton-Jacobi-Isaac partial differential equations. 

Other logics dealing with multi-player settings are e.g., coalition logic \cite{pauly2002}, ATL \cite{Alur2002} or stochastic game logic \cite{Baier2007}. These logics are able to verify that a coalition of arbitrarily many players is able to reach a certain goal. The game that is played is defined in the semantics. That has the disadvantage that validity demands a formula to be true for all games, which is only the case for very general statements. Our logic includes the game's definition in the formula, so that safety guarantees can be made for a specific game. Furthermore, the referenced works \cite{pauly2002, Alur2002, Baier2007} only consider reachability of outcomes, whereas our work includes goals for all players.

Agotnes et al. \cite{agotnes2006} include a preference relation in their logic but keep reachability of outcomes and their preferability separate which leaves information unused. Moreover, games are not composable whereas our logic allows modularizing games so that proofs can be reused. Additionally, the players in our work take each other's goals into account to improve their strategies.

Closest to our work are the propositional game logic developed by Parikh \cite{parikh85} which addresses zero-sum two-player games, and differential game logic \dGL by Platzer \cite{Platzer15} which extends game logic to include continuous dynamics. Relations to the $\mu$-calculus are discussed in prior work \cite{DBLP:conf/lics/AbouElWafaP24}. Our work takes these logics a step further by allowing non-zero sum games instead.

\section{Preliminaries}
\label{preliminaries}

For better understanding of the content of the paper, the logic \dGL \cite{Platzer15} will be recalled briefly in this section. In \dGL, there are two players called \emph{Angel} and \emph{Demon}. These players play a \emph{hybrid game} which is specified as part of a logical formula. In the following, syntax, semantics and a proof calculus for \dGL will briefly be explained, based on the book \cite{Platzer2018}.

The syntax of \dGL is    based on first-order logic. Additionally, there are two modalities $\angel{\alpha}{P}$ and $\demon{\alpha}{P}$. The first one states that Angel can win the game $\alpha$ by achieving her goal $P$. The second modality means that Demon can win game $\alpha$ by reaching his goal $P$. Formally, the syntax is defined as follows:

    \begin{table}[tbhp]
    \caption{Hybrid games}
    \label{tab:games}
    \begin{tabularx}{\textwidth}{|l| l |X|}
        \hline
        \textbf{Game} & \textbf{Name} & \textbf{Meaning}\\
        \hline
        $x:=e$ & Assignment game & assigns $e$ to $x$\\
        $x'=f(x)\&Q$ & Continuous game & Angel evolves ordinary differential equation (ODE) to change value of $x$ while evolution domain constraint $Q$ has to hold\\
        $?Q$ & Test game & tests if Angel fulfills $Q$, if not, she loses and Demon wins automatically\\
        $\alpha \cup \beta$ & Choice game & Angel chooses to play either $\alpha$ or $\beta$\\
        $\alpha; \beta$ & Sequential game & $\alpha$ and $\beta$ are played sequentially\\
        $\alpha^d$ & Dual game & controls in $\alpha$ are swapped between Angel and Demon\\
        $\alpha^*$ & Repetition game & $\alpha$ is played repeatedly until Angel wants to stop after finite rounds\\
        \hline 
    \end{tabularx}
    \end{table}

\begin{definition}[\dGL Syntax]
  Formulas of \dGL are defined by the grammar
    \[P,Q ::= e\geq \Tilde{e} \;|\; \lnot P \;|\; P\land Q \;|\; \forall x P \;|\; \exists x P \;|\; \angel{\alpha}{P} \;|\; \demon{\alpha}{P}\]
    where $x$ is a variable, $e, \tilde{e}$ terms, $P,Q$ formulas and $\alpha$ a hybrid game (\rref{tab:games}).
\end{definition}
The formulas are interpreted over states. Each state is a function $\omega: \mathcal{V} \to \mathbb{R}$ which assigns a real value to each variable. The variable values in state $\omega^r_x$ are the same as in $\omega$, except for $x$ whose value is $r$. The semantics is defined as a function $\winr{\cdot}: Fml \to \mathcal{P}(\s)$ which returns the states where a formula is true. The semantics for the first-order formulas is as usual. The semantics for the two modalities is defined using two functions that retrace from which states a player must have started to reach their goal at the end of the game.

\begin{definition}{\emph{(Semantics)}}
    The \dGL semantics is:
    \begin{itemize}
        \item $\winr{\angel{\alpha}{P}} = \winrA{\alpha}{\winr{P}}$
        \item $\winr{\demon{\alpha}{Q}} = \winrD{\alpha}{\winr{Q}}$
    \end{itemize}
    Angel's function for her winning region $\varsigma_\alpha(\cdot)$ is defined as follows:
    \begin{itemize}
        \item $\winrA{x:=e}{X} = \{\omega\in \s \;|\; \omega_x^{\omega\winr{e}} \in X\}$
    \item $\winrA{x'=f(x)\&Q}{X} = \{\varphi(0) \in \s \;|\; \varphi(r)\in X \text{ for some } r \geq 0 \text{ and (differentiable) }\allowbreak \varphi:[0,r] \to  \s \text{ such that } \varphi(s) \in \winr{Q} \text{ and } \frac{d\varphi(t)(x)}{dt}(s) = \varphi(s)\winr{f(x)} \text{ for all }\allowbreak 0 \leq s \leq r\} \stackrel{\text{def}}{=} \{\varphi(0) \in \s \;|\; \varphi(r)\in X \text{ for some } r \geq 0 \text{ with }\varphi \models x'=f(x) \land Q\}$
        \item $\winrA{?Q}{X} = \winr{Q} \cap X$
        \item $\winrA{\alpha\cup \beta}{X} = \winrA{\alpha}{X}\cup \winrA{\beta}{X}$
        \item $\winrA{\alpha;\beta}{X} = \winrA{\alpha}{\winrA{\beta}{X}}$
        \item $\winrA{\alpha^d}{X} = \winrA{\alpha}{X^\C}^\C$
        \item $\winrA{\alpha^*}{X} = \bigcap\{Z\subseteq \s\;|\; X\cup\winrA{\alpha}{Z} \subseteq Z\}$
    \end{itemize}
    Demon's winning region is defined by $\winrD{\alpha}{X} = \winrA{\alpha}{X^\C}^\C$ (\cite[Th. 3.1]{Platzer15}), i.e., Demon wins whenever Angel fails to reach the opposite of his goal, because he wins whenever Angel loses.
\end{definition}
The proof calculus for \dGL shown in Table \ref{tab:dgl_calculus}, consists of all proof rules for first-order logic and one axiom for each game in Angel's modality. Additionally, there is also a determinacy axiom that links Angel's modality to Demon's. All axioms for Demon's modality can be derived using this axiom. Furthermore, the calculus contains a monotonicity rule and a fixpoint rule that handles repetition games. 
\renewcommand{\arraystretch}{1.1}
\begin{table}[htb]
    \caption{\dGL proof calculus}
    \label{tab:dgl_calculus}
    \begin{tabularx}{\textwidth}{>{\color{gray}}l X >{\color{gray}}l l}
    $[\cdot]$ & $\demon{\alpha}{P} \leftrightarrow \lnot \angel{\alpha}{\lnot P}$ & $\langle \cup \rangle$ & $\angel{\alpha \cup \beta}{P} \leftrightarrow \angel{\alpha}{P} \lor \angel{\beta}{P}$ \\
    $\langle := \rangle$ & $\angel{x:=e}{p(x)} \leftrightarrow p(e)$ & $\langle {;} \rangle$ & $\angel{\alpha;\beta}{P} \leftrightarrow \angel{\alpha}{\angel{\beta{P}}}$ \\
    $\langle ' \rangle$ &  $\angel{x'=f(x)}{P} \leftrightarrow \exists t{\geq}0\, \angel{x:=y(t)}{P}$ \hskip 7pt $(y'= f(y))$ & $\langle^d\rangle$ & $\angel{\alpha^d}{P} \leftrightarrow \lnot \angel{\alpha}{\lnot P}$ \\
    $\langle?\rangle$ & $\angel{?Q}{P} \leftrightarrow Q \land P$ & $\langle^*\rangle$ & $\angel{\alpha^*}{P} \leftrightarrow P \lor \angel{\alpha}{\angel{\alpha^*}{P}}$\\
    M &
    \AxiomC{$P \rightarrow Q$}
    \UnaryInfC{$\angel{\alpha}{P} \rightarrow \angel{\alpha}{Q}$}
    \DisplayProof 
     & FP &
     \AxiomC{$P \lor \angel{\alpha}{Q} \rightarrow Q$}
    \UnaryInfC{$\angel{\alpha^*}{P} \rightarrow Q$}
    \DisplayProof
\end{tabularx}
\end{table}
\section{Semi-Competitive Differential Game Logic}
\label{main}
The logic \dGLsc supports two players called \emph{Angel} and \emph{Demon}. 
They semi-competitively play a game defined in the \dGLsc formula.

After the game has ended, \emph{both} want to have achieved their goal. Each player has a separate goal which is independent of the other player's goal. The game is \emph{non-zero sum}, both players can win or lose at the same time.
\subsection{Semi-Competitiveness}
To describe the player's behavior in the logic, the notion of ``semi-competitive\-ness'' is introduced in this paper. 
\emph{Semi-competitiveness means that players cooperate where it is helpful for them, and compete where it is necessary}.
More precisely, a player will only choose options that help them reach their goal. They will never lose on purpose. If there are multiple options that all make the player win, the player will \emph{always} choose an option that also helps the other player win. If no option makes the player win, they will not altruistically help the other player instead. In terms of game theory, the players play a cooperative game, if collaboration is possible, and a non-cooperative game otherwise.

This gameplay is chosen to ensure the greatest cooperation among players, especially in games with multiple steps. Here, uncooperative behavior in a previous step might have negative consequences later. For example, if Angel and Demon pick candy for each other and Angel gives Demon a lemon candy, although he wanted a strawberry candy, he might in turn give Angel a candy she does not want because he already lost the game and will not help Angel anymore. If, on the other hand, Angel indeed gives him a strawberry candy, Demon will give her the candy she wants because now all options make him win, so he will help Angel. 

Additionally, this gameplay ensures monotonicity of the logic, i.e., a greater goal means more possibilities to win. Monotonicity ensures that the logic behaves intuitively. Furthermore, it guarantees the existence of fixpoints which are needed later on to define parts of the semantics. Helping other players is a necessary assumption to obtain monotonicity: If the players' behavior does not also help the other player win and players choose randomly between all options they are indifferent to, then monotonicity is lost.

\subsection{Syntax}
The syntax of \dGLsc consists of first-order logic and two modalities. The modality $\angel{\alpha}{P}{Q}$ expresses that Angel has a winning strategy to win the game $\alpha$ by achieving her goal $P$ while knowing that Demon's goal is $Q$. $\demon{\alpha}{P}{Q}$ means that Demon has a winning strategy  to achieve $Q$ by playing $\alpha$ while knowing that Angel's goal is $P$.
\begin{definition}[\dGLsc Syntax]
    \label{def:dglsc_syntax}
The following grammar defines the \dGLsc formulas:
\begin{itemize}
    \item $\alpha, \beta ::= x:=e \;|\; x'=f(x) \& Q \;|\; ?Q \;|\; \alpha \cup \beta \;|\; \alpha ; \beta \;|\; \alpha^d \;|\; \alpha ^*$
    \item $P, Q ::= e \geq \Tilde{e} \;|\; \lnot P \;|\; P \land Q \;|\; \forall x P \;|\; \exists x P \;|\; \angel{\alpha}{P}{Q} \;|\; [\alpha](P,Q)$
\end{itemize}
where $\alpha, \beta$ are hybrid games, $P, Q$ are \dGLsc formulas, $ x$ is a variable, $f$ is a function and $e, \Tilde{e}$ are terms.
\end{definition}
The hybrid games in Definition \ref{def:dglsc_syntax} have the same effects as the ones for \dGL described in Table \ref{tab:games}. So in the game $((x:=x+1 \cup x:=x-1); \{x'= -1\}^d)^*$, for example, the outermost game is a repetition game. Here, Angel chooses after each round, if she wants to play again. She might also choose to play 0 rounds. If she does play a round, Angel first chooses whether to increase or to decrease $x$ by one. After that follows a continuous game. As it is marked with $^d$, this game is under Demon's control, so he chooses how long to evolve the ODE to decrease $x$. That means, time passes and $x$ changes according to the ODE until Demon stops time. In this example, the new value of $x$ would be $x-t$ after time $t$. Then, Angel chooses again whether she wants to play another round or not. 

The main difference now to \dGL is not in the effects of games but rather in the behavior of the players: If we have $\angel{(x:=1 \cup x:= 0)^d}{x=1}$ in \dGL, Demon will choose an option that hinders Angel at achieving her goal due to his adversarial behavior, i.e., $x:= 0$ here. In \dGLsc the formula could be $\angel{(x:=1 \cup x:= 0)^d}{x=1}{\top}$, including the information that Demon's goal is \emph{true}. As Demon can win in any case, Demon helps Angel by choosing $x:=1$ due to his semi-competitive behavior. 

\begin{example}
    \label{kayak}
    As the verification of an aerial collision avoidance system is a fully-fledged case study, the two planes in this example are still standing on the ground.
    Angel and Demon are replenishing the stocks for their respective plane, but Angel is missing 3l of orange juice while Demon needs 5l of tomato juice (more would not fit in the tank, less and some passengers will stay thirsty). 
    Fortunately, Angel has plenty of tomato juice and Demon's stock is full of orange juice (enough, so that they do not have to worry about running out).
    In \dGLsc, filling up each other's stock can be modeled using two continuous games:
    \[o=0 \land t=0 \rightarrow \angel{\{t'=1\}; \{o'=1\}^d}{o=3}{t=5}
    \]
    First, Angel, at a constant rate, fills tomato juice into Demon's tank.
    Then, Demon fills orange juice into her tank.
    Note how the goals of both players are visible at the end of the formula.
    Although these goals are not directly overlapping, the players' semi-competitive behavior ensures their cooperation which leads to them both winning the game. 
    This would not have been possible if they played adversarially because they have no control over the variable that matters for their goal.
    Even with indifference towards Demon's goal, Angel could not have won the game.
    If she would not care about the amount of tomato juice she gave to Demon, Demon would probably have refused to cooperate with her, making her lose as well.
\end{example}
\subsection{Semantics}
The semantics of the formulas is defined in the following section. 
This is done via a function $\winr{\cdot}: \textit{Fml} \to \mathcal{P}(\s)$ that maps a formula to the set of states where the formula is true. 
$\s$ is the set of all states. A state is a mapping $\omega: \mathcal{V} \to \mathbb{R}$ that maps all variables to a real number. 
The state $\omega_x^r$ denotes a state that coincides with the state $\omega$ in all variable values, except for $x$ whose value is replaced by $r$.

The semantics of first-order operators is as expected. 
To define the semantics of Angel's and Demon's modality, the functions $\winrA{\alpha}{\cdot}{\cdot}$ and $\winrD{\alpha}{\cdot}{\cdot}$ resp., are used which describe the winning region of a game. 
These functions take two inputs each: The first one is Angel's goal and the second is Demon's goal. Since Angel and Demon know each other's goal, this knowledge is needed for determining the winning regions.
Unlike in \dGL, Demon's semantics cannot be defined as a function of Angel's winning region $\winrA{\alpha}{\cdot}{\cdot}$ anymore, because their winning regions are no longer complementary. 
Instead, their winning regions partially overlap and there are also states that belong to neither winning region. 

Still, the semantics of the assignment, test, and sequential game is similar to \dGL as the effects of the game are the same, and the players' decisions cannot differ from those they would make in \dGL because these games do not involve choices.
A little tweak for the semantics of the sequential game is needed nonetheless, as the winning regions in \dGLsc take two arguments instead of one. 
To win the sequential game $\alpha;\beta$, the player has to reach their winning region for $\beta$ after $\alpha$, so they can reach their goal after $\beta$. The same applies for the other player except their respective winning region operator is needed, instead, so the semantics of the sequential game is:
\[\winrA{\alpha;\beta}{X}{Y} = \winrA{\alpha}{\winrA{\beta}{X}{Y}}{\winrD{\beta}{X}{Y}} \text{ and } \winrD{\alpha;\beta}{X}{Y} = \winrD{\alpha}{\winrA{\beta}{X}{Y}}{\winrD{\beta}{X}{Y}}\]

For Angel, the semantics of the continuous game and the choice game also correspond to the semantics in \dGL.
The winning region only returns states in which Angel \emph{has} an option so that she can win but which one she will \emph{actually} choose is invisible.
Locally, only her goal matters to her because she has control which is also why she cannot expect help from Demon in this step. Globally, semi-competitiveness comes into play through nested games, e.g., the games Angel chooses from in the choice game. Especially important here is the dual game because in this game Demon's goals matter for Angel which then inductively matter for Angel through the whole game.

Demon's semantics of the continuous game and choice game does differ from the one in \dGL though. 
His winning regions are extended by states where Angel, who is in control, will help him. As they act semi-competitively, Angel will help him in states where both of their goals can be reached simultaneously. 
Therefore, Demon's semantics for the continuous game \(\winrD{x'=f(x)\&Q}{X}{Y}\) is:
\[\begin{aligned}[t]
    &\{\varphi(0) \in \s\;|\; \varphi(r) \in Y \text{ for all } r \text{ with } \varphi \models x'=f(x) \land Q\}\\
    &\cup \begin{aligned}[t]
        \{\varphi(0) \in \s \;|\; \varphi(r) \in X \cap Y ~
        \text{ for some } r \text{ with } \varphi \models x'=f(x) \land Q\}&
    \end{aligned}
\end{aligned}\]
The definition of the semantics uses a function from time to states $\varphi$. 
Over the course of this function, $x$ changes such that $x' = f(x)$ and the evolution domain constraint $Q$ holds in every state in $\varphi$. 
If Demon stays in his goal $Y$ for all states in $\varphi$, he can definitely win the game (left-hand side of the union). 
If Angel can reach a state by evolving the ODE (i.e. there is a state in $\varphi$) where both of their goals are fulfilled, she will stop there because she behaves semi-competitively, making both of them win (right-hand side of the union).
The definition for the semantics of the choice game follows a corresponding idea:
\begin{multline*}\winrD{\alpha\cup \beta}{X}{Y} = (\winrD{\alpha}{X}{Y} \cap \winrD{\beta}{X}{Y}) \cup (\winrD{\alpha}{X}{Y} \cap \winrA{\alpha}{X}{Y})\\ \cup (\winrD{\beta}{X}{Y} \cap \winrA{\beta}{X}{Y})
\end{multline*}
If Demon can win both $\alpha$ and $\beta$, i.e. he is in the winning region of $\alpha$ and $\beta$, he can definitely win the game.
If he and Angel can both win $\alpha$, i.e. the current state is part of both Angel's and Demon's winning region for $\alpha$, Demon can also win the game because Angel will choose $\alpha$ due to her semi-competitive behavior.
A similar argument applies, if both of them can win $\beta$.

In the dual game, the controls between Angel and Demon are swapped. 
This is simulated in the semantics by letting Angel keep control but making her play for Demon, i.e. trying to achieve his goal.
In exchange, Demon plays with Angel's goal. 
As Angel now wins, if Demon's goal is achieved, her winning region actually denotes the states where Demon wins, as well as vice versa.
Consequently, Angel and Demon swap goals and winning regions in the definitions of the dual game:
\[\winrA{\alpha^d}{X}{Y} = \winrD{\alpha}{Y}{X} \text{ and } \winrD{\alpha^d}{X}{Y} = \winrA{\alpha}{Y}{X}\]

The main idea to define the semantics of the repetition game is that allowing one more round of $\alpha$ should not change the winning region. Consequently, the winning region for Angel should be some kind of fixpoint of the form $X \cup \winrA{\alpha}{\winrA{\alpha^*}{X}{Y}}{\winrD{\alpha^*}{X}{Y}} = \winrA{\alpha^*}{X}{Y}$ as Angel either is already in her goal, or she can reach it after one more round of $\alpha$. The problem is that by the same idea, Demon's winning region should also be a fixpoint, so using a definition like this would result in complex nested fixpoints. Additionally, this fixpoint is not unique yet. To solve the first problem, a case distinction is made between competitive and cooperative behavior, resulting in two fixpoints. If they compete, the other's goal is the complement of the fixpoint because they have complementary goals, while in the cooperative case, the other's goal is the same fixpoint, as they share the same goal. Consequently, a candidate $Z$ for Angel's first fixpoint should fulfill $X\cup \winrA{\alpha}{Z}{Z^\C} = Z$ and for the second it should fulfill $(X\cap Y) \cup (\winrA{\alpha}{Z}{Z} \cap \winrD{\alpha}{Z}{Z}) = Z$. The intersection of their winning regions rules out that one player wins and fulfills both goals (otherwise they would not cooperate), but the other one loses, e.g., by failing a test. To make the fixpoints unique, either the greatest or the least fixpoint can be chosen. $\s$ as the greatest possible set fulfills the equation, but this does not contain much information, so the least fixpoint is used. It can be computed as the intersection of all pre-fixpoints (sets that fulfill the equation only with set inclusion). Taken together, the semantics for Angel is:
\[\winrA{\alpha^*}{X}{Y} = \begin{aligned}[t]
&\bigcap\{Z\subseteq \s \;|\; X \cup \winrA{\alpha}{Z}{Z^\C} \subseteq Z\}\\
&\cup \bigcap\{Z \subseteq \s \;|\; (X\cap Y) \cup (\winrA{\alpha}{Z}{Z} \cap \winrD{\alpha}{Z}{Z})\}
\end{aligned}
\]
The definition of Demon's winning region follows the same idea. The second fixpoint is identical because here, both of them win. If Demon competes with Angel, he can only win for sure if he always stays in his goal, no matter how many rounds Angel plays. Consequently, Demon needs to be in his goal now, and he should be in his winning region, even if one more round of $\alpha$ is allowed: $Z = Y \cap \winrD{\alpha}{Z^\C}{Z}$. Because the least fixpoint in this case is the empty set, which does not contain much information, the greatest fixpoint is used. It can be computed by union of all pre-fixpoints. Thus, the semantics for Demon are:
\[\winrD{\alpha}{X}{Y} = \begin{aligned}[t]
&\bigcup\{Z\subseteq \s \;|\; Z \subseteq Y \cap \winrD{\alpha}{Z^\C}{Z}\}\\ 
&\cup \bigcap\{Z \subseteq \s \;|\; (X\cap Y) \cup (\winrA{\alpha}{Z}{Z} \cap \winrD{\alpha}{Z}{Z})\}
\end{aligned}
\]

The full definition of the semantics can be found in the \appendixLocation.

\begin{remark}
    As in \dGL, goals can still be converted to tests at the end, i.e. there is an equivalent formulation for each modality where the goals are replaced by \emph{true} and the original goals appear as tests at the end of the program. As Angel's and Demon's winning regions are not complementary anymore, this conversion is much more subtle than in \dGL because Angel does not always win if Demon loses. Additionally, Demon's ability to win influences his choices and thus Angel's ability to win and vice versa. Therefore, tests for both goals must be incorporated, and Angel and Demon, resp., need to pass them in any order (either test Angel first, or test Demon first) to ensure that they do not only win because the other one loses. This also demonstrates that the opponent's goals matter to both players.
    For the full statement, reference \appendixLocationTwo.
\end{remark}

\section{\dGLsc Properties and Relation to Zero-sum \dGL}
\label{properties}
In this section, some important results for \dGLsc are proven. The first result proves monotonicity which ensures that the logic behaves in an intuitive way: A larger goal can be reached from a larger winning region.
The first lemma states that the logic is monotone if both goals are increased. In this case, both Angel's and Demon's winning region will increase.
The proof for this lemma is a straightforward structural induction over $\alpha$.

\begin{lemmaE}[Monotonicity][end]
\label{monotonicity}
    The logic \dGLsc is monotone, i.e.
    \[
    X\subseteq A, Y\subseteq B \text{ implies } \winrA{\alpha}{X}{Y} \subseteq \winrA{\alpha}{A}{B} \text{ and } \winrD{\alpha}{X}{Y} \subseteq \winrD{\alpha}{A}{B}
    \]
\end{lemmaE}
\begin{proofE}
    The proof is conducted by induction over the structure of the game $\alpha$.
    \begin{enumerate}
        \item For the atomic games (assignment, test and continuous game), $\winrA{\alpha}{X}{Y} \subseteq \winrA{\alpha}{A}{B}$ and $\winrD{\alpha}{X}{Y} \subseteq \winrD{\alpha}{A}{B}$ directly follows from replacing $X$ and $Y$ in the definition of the winning region by the bigger sets $A$ and $B$ respectively.
        \item Similarly, monotonicity can be proved for the repetition game.
        \item For the choice game and the dual game, monotonicity follows from applying the IH to all parts in their winning region's definition (as they solely consist of winning regions of games with lower structural complexity).
        \item $\winrA{\alpha;\beta}{X}{Y} = \winrA{\alpha}{\winrA{\beta}{X}{Y}}{\winrD{\beta}{X}{Y}}$ by definition. As $\winrA{\beta}{X}{Y} \subseteq \winrA{\beta}{A}{B}$ and $\winrD{\beta}{X}{Y} \subseteq \winrD{\beta}{A}{B}$ by IH, the IH can be applied, yielding \[\winrA{\alpha}{\winrA{\beta}{X}{Y}}{\winrD{\beta}{X}{Y}} \subseteq \winrA{\alpha}{\winrA{\beta}{A}{B}}{\winrD{\beta}{A}{B}} = \winrA{\alpha;\beta}{A}{B}\] $\winrD{\alpha;\beta}{X}{Y} \subseteq \winrD{\alpha;\beta}{A}{B}$ is proved similarly.
    \end{enumerate}
    \qed
\end{proofE}

The second lemma states that if Angel's and Demon's goals are disjoint, Angel's winning region expands by increasing her goal while Demon's goal may change arbitrarily. Intuitively, all parts in the semantics covering cooperation are empty since the goals are disjoint. Then, proving that the remaining sets grow by increasing the goal is fairly easy. After that, the parts where Angel and Demon cooperate, can be added. Thus, the winning region can only increase. The proof for Demon's statement is conducted similarly.

The following lemma is especially helpful if Angel and Demon have opposing goals, before and after the increase of one goal. In this case, enlarging one goal means that the other shrinks. Therefore, the first monotonicity lemma \ref{monotonicity} is not applicable but Lemma \ref{mon2} can be used.
\begin{lemmaE}[Disjoint Monotonicity][end]
    \label{mon2}
    If the goals of Angel and Demon are initially disjoint, then \dGLsc is monotone in each argument:
    \[
    \begin{aligned}
        &X\subseteq A, X\cap Y = \emptyset \text{ implies } \winrA{\alpha}{X}{Y} \subseteq \winrA{\alpha}{A}{B}
        \text{  and}\\
        &Y \subseteq B, X \cap Y = \emptyset \text{ implies } \winrD{\alpha}{X}{Y} \subseteq \winrD{\alpha}{A}{B}
    \end{aligned}
    \]
    Sets $A$ and $B$ may overlap. For each case the other goal can be chosen arbitrarily.
\end{lemmaE}
\begin{proofE}
    The proof is conducted by induction over the structure of the game $\alpha$ simultaneously for both claims.
    \begin{enumerate}
        \item For the assignment game and the test game, $\winrA{\alpha}{X}{Y} \subseteq \winrA{\alpha}{A}{B}$ and $\winrD{\alpha}{X}{Y} \subseteq \winrD{\alpha}{A}{B}$ directly follows from replacing $X$ and $Y$ resp. by the bigger sets $A$ and $B$ resp. in the winning region's definition.
        \item For the continuous game, this works as well to prove $\winrA{x'=f(x)\& Q}{X}{Y} \subseteq \winrA{x'=f(x)\& Q}{A}{B}$. 
        
        To prove $\winrD{x'=f(x)\& Q}{X}{Y} \subseteq \winrD{x'=f(x)\& Q}{A}{B}$, the assumption that $X\cap Y = \emptyset$ is also necessary. By definition,
        \[\winrD{x'=f(x)\& Q}{X}{Y} = \begin{aligned}[t]
            &\{\varphi(0) \in \s \;|\; \varphi(r) \in Y \text{ for all }r \text{ with } \varphi \models f(x) \land Q\}\\
            &\cup \{\varphi(0) \in \s \;|\; \varphi(r) \in (X\cap Y) \text{ for some }r \text{ with } \varphi \models f(x) \land Q\}
        \end{aligned}
        \]
        As $X\cap Y$ is empty, the union's right-hand side is empty. The remaining set can be increased by replacing $Y$ with the bigger set $B$, yielding $\{\varphi(0) \in \s \;|\; \varphi(r) \in B \text{ for all }r \text{ with } \varphi \models f(x) \land Q\}$. This set is then further increased by adding $\{\varphi(0) \in \s \;|\; \varphi(r) \in (A\cap B) \text{ for some }r \text{ with } \varphi \models f(x) \land Q\}$ to obtain the definition of $\winrD{x'=f(x)\& Q}{A}{B}$.
        \item $\winrA{\alpha\cup\beta}{X}{Y} \subseteq\winrA{\alpha\cup\beta}{A}{B}$ can directly be proven by applying the IH to all parts of the winning region's definition.
        
        By definition,
        \[\winrD{\alpha\cup\beta}{X}{Y} = \begin{aligned}[t]
            &(\winrD{\alpha}{X}{Y} \cap \winrD\beta{X}{Y})\\
            &\cup (\winrD{\alpha}{X}{Y}\cap\winrA{\alpha}{X}{Y})\\
            &\cup (\winrD{\beta}{X}{Y} \cap \winrA{\beta}{X}{Y})\\
        \end{aligned}\]
        As $X\cap Y$ is empty, Angel and Demon cannot win at the same time. Therefore, all parts except $\winrD{\alpha}{X}{Y} \cap \winrD\beta{X}{Y}$ are empty. By IH, this remaining intersection is a subset of $\winrD{\alpha}{A}{B} \cap \winrD\beta{A}{B}$. Finally, $\winrD{\alpha\cup\beta}{A}{B}$ is obtained by adding $\winrD{\alpha}{A}{B}\cap\winrA{\alpha}{A}{B}$ and $\winrD{\beta}{A}{B} \cap \winrA{\beta}{A}{B}$.
        \item By definition, $\winrA{\alpha;\beta}{X}{Y} = \winrA{\alpha}{\winrA{\beta}{X}{Y}}{\winrD{\beta}{X}{Y}}$. As $X\cap Y$ is empty, $\winrA{\beta}{X}{Y}\cap \winrD{\beta}{X}{Y}$ must also be empty because Angel and Demon cannot win at the same time. Additionally, $\winrA{\beta}{X}{Y} \subseteq \winrA{\beta}{A}{B}$ by IH. Consequently, the IH can also be applied to $\winrA{\alpha}{\winrA{\beta}{X}{Y}}{\winrD{\beta}{X}{Y}}$, yielding the bigger set $\winrA{\alpha}{\winrA{\beta}{A}{B}}{\winrD{\beta}{A}{B}} = \winrA{\alpha;\beta}{X}{Y}$. The case 
        for Demon is proved in the same way.
        \item For the dual game, both claims directly follow from applying the IH to the winning region's definition.
        \item By definition, \[\winrA{\alpha^*}{X}{Y} = \begin{aligned}
            &\bigcap\{Z\subseteq \s \;|\; X \cup \winrA{\alpha}{Z}{Z^\C} \subseteq Z\}\\
            &\cup \bigcap\{Z\subseteq \s \;|\; (X\cap Y) \cup (\winrA{\alpha}{Z}{Z}\cap \winrD{\alpha}{Z}{Z}) \subseteq Z\}
        \end{aligned}\]
        As $X\cap Y$ is empty, the right-hand side of the union is empty. The remaining set can be increased by replacing $X$ with the bigger set $A$, yielding $\bigcap\{Z\subseteq \s \;|\; A \cup \winrA{\alpha}{Z}{Z^\C} \subseteq Z\}$. This set is then further increased by adding $\bigcap\{Z\subseteq \s \;|\; (A\cap B) \cup (\winrA{\alpha}{Z}{Z}\cap \winrD{\alpha}{Z}{Z}) \subseteq Z\}$ to obtain the definition of $\winrA{\alpha^*}{A}{B}$. Similarly, $\winrD{\alpha^*}{X}{Y} \subseteq \winrD{\alpha^*}{A}{B}$ can be proved.
    \end{enumerate}
    \qed
\end{proofE}

Another important property is that the logic behaves as intended. But how can this be checked? If the players have complementary goals, they have to compete and will behave adversarially. The same setting can also be expressed in \dGL (which is assumed to be defined suitably). Consequently, both formulas should also have the same truth values. This holds indeed, as Lemma \ref{CollapseDGL} shows:

\begin{lemmaE}[Complementary \dGL Equivalence][end]
    \label{CollapseDGL}
If Angel and Demon have complementary goals, the logic \dGLsc reduces to \dGL, i.e. it holds that
\[
\winrA{\alpha}{X}{X^\C} = \winrA{\alpha}{X} \quad\text{and}\quad
\winrD{\alpha}{X^\C}{X} = \winrD{\alpha}{X}
\]
\end{lemmaE}
\begin{proofE}
    The proof is conducted by induction over the structure of the game $\alpha$.
    \begin{enumerate}
        \item For the assignment game and test game for Angel and Demon, and the continuous game for Angel, the winning region's definition in \dGLsc corresponds to the definition in \dGL.
        \item By definition, \[\winrD{x'=f(x)\&Q}{X,X^\C} = \begin{aligned}[t]
            &\begin{aligned}[t]
                \{\varphi(0)\in\s\;|\; \varphi(r) \in X^\C \text{ for all } r\hspace{3cm}&\\ \text{ with } \varphi \models x'=f(x) \land Q\}&
            \end{aligned}\\
            &\cup\begin{aligned}[t]
                \{\varphi(0) \in \s \;|\; \varphi(r) \in (X\cap X^\C) \text{ for some } r \hspace{1.3cm}&\\
                \text{ with } \varphi \models x'=f(x) \land Q\}&
            \end{aligned} 
        \end{aligned}\]
        As $X\cap X^\C$ is empty, the right-hand side of the union is empty. The remaining part corresponds to $\winrD{x'=f(x)\&Q}{X^\C}$.
        \item For the choice game, $\winrA{\alpha\cup\beta}{X}{X^\C} = \winrA{\alpha\cup\beta}{X}$ can directly be proved by applying the IH to all parts of the winning region's definition.
        
        By definition, 
        \[\winrD{\alpha\cup\beta}{X}{X^\C} = \begin{aligned}[t]
            &(\winrD{\alpha}{X}{X^\C} \cap \winrD{\beta}{X}{X^\C})\\
            &\cup (\winrD{\alpha}{X}{X^\C} \cap \winrA{\alpha}{X}{X^\C})\\
            &\cap (\winrD{\beta}{X}{X^\C} \cap \winrA{\beta}{X}{X^C})
        \end{aligned}\]
        As $X\cap X^\C$ is empty, all parts except $\winrD{\alpha}{X}{X^\C} \cap \winrD{\beta}{X}{X^\C}$ are empty because Angel and Demon cannot win simultaneously. By applying the IH to the remaining parts, the definition of $\winrD{\alpha\cup \beta}{X^\C}$ is obtained.
        \item $\winrA{\alpha;\beta}{X}{X^\C} = \winrA{\alpha}{\winrA{\beta}{X}{X^\C}}{\winrD{\beta}{X}{X^\C}}$ by definition. Using the IH on $\winrA{\beta}{X}{X^\C}$ and $\winrD{\beta}{X}{X^\C}$ yields $\winrA{\alpha}{\winrA{\beta}{X}}{\winrD{\beta}{X^\C}}$. As determinacy holds for \dGL, $\winrD{\beta}{X^\C} = \winrA{\beta}{X}^\C$. Now, the IH can be applied a second time, yielding $\winrA{\alpha}{\winrA{\beta}{X}} = \winrA{\alpha;\beta}{X}$. Similarly, $\winrD{\alpha;\beta}{X}{X^\C} = \winrD{\alpha;\beta}{X^\C}$ can be proved.
        \item By definition, $\winrA{\alpha^d}{X}{X^\C} = \winrD{\alpha}{X^\C}{X}$. By IH, this is equal to $\winrD{\alpha}{X}$. Using \dGL's determinacy, this can be transformed to $(\winrA{\alpha}{X^\C})^\C = \winrA{\alpha^d}{X}$. Demon's case is proved similarly.
        \item By definition,
        \[\winrA{\alpha^*}{X}{X^\C} = \begin{aligned}[t]
            &\bigcap\{Z\subseteq \s \;|\; X \cup \winrA{\alpha}{Z}{Z^\C} \subseteq Z\}\\
            &\cup \bigcap\{Z\subseteq \s \;|\; (X\cap X^\C) \cup (\winrA{\alpha}{Z}{Z} \cap \winrD{\alpha}{Z}{Z}) \subseteq Z\}
        \end{aligned}\]
        As $X\cap X^\C$ is empty, the right-hand side of the union is empty. Applying the IH to the remaining part yields $\bigcap\{Z\subseteq \s \;|\; X \cup \winrA{\alpha}{Z} \subseteq Z\} = \winrA{\alpha^*}{X}$. Similarly, $\winrD{\alpha^*}{X}{X^\C} = \winrD{\alpha^*}{X^\C}$ is proved.
    \end{enumerate}
    \qed
\end{proofE}

In the proof it is shown via structural induction that all parts in the semantics that differ between \dGLsc and \dGL cancel out, if the goals oppose.
\begin{remark}
    If there were a super-logic that contained both \dGLsc and \dGL, the previous lemma could be rendered as the syntactic equivalences
    \[\angel{\alpha}{P}{\lnot P} \leftrightarrow \angel{\alpha}{P} \quad\text{and}\quad \demon{\alpha}{\lnot P}{P} \leftrightarrow \demon{\alpha}{P}\]
\end{remark}

\begin{notation}
    In the following, $\angel{\alpha}{P}{\lnot P}$ will be abbreviated with $\angel{\alpha}{P}$ and $\demon{\alpha}{\lnot Q}{Q}$ with $\demon{\alpha}{Q}$, as well as $\winrA{\alpha}{X}{X^\C}$ with $\winrA{\alpha}{X}$ and $\winrD{\alpha}{Y^\C}{Y}$ with $\winrD{\alpha}{Y}$. These two formulas express the same and have the same winning regions in \dGLsc and \dGL respectively(Lem.~\ref{CollapseDGL}). The abbreviations allow for staying in the logic \dGLsc without having to write $P, Q, X, Y$ twice.
\end{notation}

Consequently, determinacy holds even though it only does for complementary goals, which will be useful for the proof calculus later on.
\begin{corollaryE}[Complementary Determinacy][end]
\label{determinancy}
    Determinacy holds if Angel and Demon have complementary goals, i.e. the following equivalence is valid:
    \[
    \lnot \angel{\alpha}{P}{\lnot P} \leftrightarrow \demon{\alpha}{P}{\lnot P}
    \]
\end{corollaryE}
\begin{proofE}
    To prove this corollary, it is shown that the winning regions of both sides of the equivalence are equal. By definition, $\winr{\lnot \angel{\alpha}{P}{\lnot P}} = \winrA{\alpha}{\winr{P}}{\winr{P}^\C}^\C$. Using L. \ref{CollapseDGL}, this can be transformed to $\winrA{\alpha}{\winr{P}}^\C$ which is equal to $\winrD{\alpha}{\winr{P}^\C}$ by \dGL's determinacy. With another use of L.\ref{CollapseDGL}, this is transformed back to \dGLsc, yielding $\winrD{\alpha}{\winr{P}}{\winr{P}^\C} = \winr{\demon{\alpha}{P}{\lnot P}}$.
    \qed
\end{proofE}

To show that \dGLsc's definition is reasonable for the general case, first observe the following: If Angel and Demon compete in a later game, they will not cooperate in an earlier game because they will not reach both goals at the same time anyway. Similarly, if Angel and Demon cooperate in a later game, they can already do so for earlier games because they can win together. Inductively, this means that Angel and Demon can already decide at the start whether to cooperate or to compete! As competing players try to achieve mutually exclusive goals, each of them might as well play against an adversarial player because the other one might hinder them at achieving their goal in order to reach their own. Two cooperating players can also be modeled as an adversarial game: The coalition is regarded as one player and tries to reach both goals, playing against an adversarial player with no control. Since one of the two cases always applies, every semi-competitive game can be split into two adversarial games. Using this alternative description, the definition of \dGLsc can be checked against \dGL again. The equivalence holds indeed, as Lemma \ref{rewriting} shows. The $^{-d}$ indicates the removal of all dual operators to transfer all controls to the coalition player. 
Although the conversion is possible, there are \emph{two} winning regions in \dGL instead of one in \dGLsc which doubles the work for proving formulas. On top, the formulas are linked by disjunction which makes it impossible to get rid of one part if each formula only holds for certain cases. Additionally, \dGLsc offers a much more intuitive way of expressing games where both players have an individual goal.

\begin{lemmaE}[General \dGL Equivalence][end]
\label{rewriting}
    Winning regions in \dGLsc correspond to those in \dGL as follows:
    \[\winrA{\alpha}{X}{Y} = \winrA{\alpha^{-d}}{X\cap Y} \cup \winrA{\alpha}{X} \quad\text{and}\quad
    \winrD{\alpha}{X}{Y} = \winrA{\alpha^{-d}}{X \cap Y} \cup \winrD{\alpha}{Y}\]
    where $\alpha$ is a hybrid game, $X,Y$ are sets of states and $\alpha^{-d}$ is defined in Table~\ref{tab:d_def}.

    \begin{table}[b]
        \caption{Hybrid system $\alpha^{-d}$ is the \emph{systematization} of hybrid game $\alpha$}
        \label{tab:d_def}
    \begin{tabularx}{\textwidth}{X X}
            $(x:=e)^{-d} \equiv x:=e$ & $(\alpha;\beta)^{-d} \equiv \alpha^{-d};\beta^{-d}$ \\
            $(x'=f(x)\& Q)^{-d} \equiv x'=f(x)\& Q$ & $(\alpha^d)^{-d} \equiv \alpha^{-d}$ \\
            $(?Q)^{-d} \equiv ?Q$ & $(\alpha^*)^{-d} \equiv (\alpha^{-d})^*$ \\
            $(\alpha \cup \beta)^{-d} \equiv \alpha^{-d} \cup \beta^{-d}$ & \\
    \end{tabularx} 
\end{table}
\end{lemmaE}
\begin{proofE}
    First, three auxiliary lemmas are proven to facilitate the proof of L. \ref{rewriting}. The first lemma (L.\ref{dgl cut rewrite}) shows that if both Angel and Demon can win against an adversarial player, they can win by forming a coalition. The second and third (L. \ref{incl_Angel} and L.\ref{incl_Demon}) consider the situation where Angel and Demon resp. play a sequential game. If they can win by playing the first game on their own and choose between cooperating or not for the second game, but cannot win both games on their own, they can also win by cooperating for both games.

    \begin{lemma}
    \label{dgl cut rewrite}
        It holds that:
        \[ \winrA{\alpha}{X} \cap \winrD{\alpha}{Y}\subseteq \winrA{\alpha^{-d}}{X\cap Y} \]
    \end{lemma}
    \begin{proof}
        The proof is conducted by induction over the structural complexity of $\alpha$.
        \begin{enumerate}
            \item $\winrA{x:=e}{X} \cap \winrD{x:=e}{Y} = \winrA{(x:=e)^{-d}}{X\cap Y}$ and $\winrA{?Q}{X} \cap \winrD{?Q}{Y} = \winrA{(?Q)^{-d}}{X\cap Y}$ can easily be seen by intersecting the winning regions' definitions.
            \item By definition \[\winrA{x'=f(x)\&Q}(X) \cap \winrD{x'=f(x) \& Q}{Y} = \begin{aligned}[t]
                &\begin{aligned}[t]
                    \{\varphi(0)\in\s\;|\; \varphi(r) \in X \text{ for some } r\hspace{1.4cm}&\\
                    \text{ with } \varphi \models x'=f(x)\land Q\}&
                \end{aligned}
                \\
                &\begin{aligned}[t]
                    \cup \{\varphi(0)\in\s\;|\; \varphi(r) \in Y \text{ for some } r \hspace{1.2cm}&\\
                    \text{ with } \varphi \models x'=f(x) \land Q \}&
                \end{aligned}
            \end{aligned}\]
            i.e. $Y$ is true in all states and there exists a state where $X$ is true. Consequently, there exists a state where both $X$ and $Y$ are true. Keeping only this requirement and dropping the requirement that $Y$ has to be true in all states, increases the set to $\{\varphi(0)\in\s\;|\; \varphi(r)\in X\cap Y \text{ for some } r \text{ with } \varphi \models x'=f(x) \land Q\} = \winrA{(x'=f(x) \&Q)^{-d}}{X\cap Y}$.
            \item By definition, 
            \[\begin{aligned}[t]\winrA{\alpha\cup\beta}{X}\cap \winrD{\alpha\cup\beta}{Y} &= 
                (\winrA{\alpha}{X} \cup \winrA{beta}{X}) \cap (\winrD{\alpha}{Y}\cap \winrD{\beta}{Y})\\
                &= (\winrA{\alpha}{X} \cap \winrD{\alpha}{Y}\cap \winrD{\beta}{Y}) \cup (\winrA{\beta}{X} \cap \winrD{\alpha}{Y}\cap \winrD{\beta}{Y})
            \end{aligned}\]
            By dropping $\winrD{\beta}{Y}$ and $\winrD{\alpha}{Y}$ resp. from the left- and right-hand side of the union resp., the set is increased to $(\winrA{\alpha}{X} \cap \winrD{\alpha}{Y}) \cup (\winrA{\beta}{X} \cap \winrD{\beta}{Y})$. Then, the IH is used on both sides of the union, yielding $\winrA{\alpha^{-d}}{X\cap Y} \cup \winrA{\beta^{-d}}{X\cap Y} = \winrA{(\alpha\cup\beta)^{-d}}{X\cap Y}$.
            \item $\winrA{\alpha;\beta}{X} \cap \winrD{\alpha;\beta}{Y} = \winrA{\alpha}{\winrA{\beta}{X}} \cap \winrD{\alpha}{\winrD{\beta}{Y}}$ by definition. First, the IH is used on the outer winning regions for $\alpha$, increasing the set to $\winrA{\alpha^{-d}}{\winrA{\beta}{X} \cap \winrD{\beta}{Y}}$. Now the IH can be used a second time on the inner winning regions for $\beta$, yielding $\winrA{\alpha^{-d}}{\winrA{\beta^{-d}{X\cap Y}}} = \winrA{(\alpha;\beta)^{-d}}{X\cap Y}$.
            \item By definition $\winrA{\alpha^d}{X} \cap \winrD{\alpha^d}{Y} = (\winrA{\alpha}{X^\C})^\C \cap (\winr{\alpha}{Y^\C})^\C$. By \dGL's determinacy, this can be transformed to $\winrD{\alpha}{X} \cap \winrA{\alpha}{Y}$. Applying the IH yields $\winrA{\alpha^{-d}}{X\cap Y} = \winrA{(\alpha^{d})^{-d}}{X\cap Y}$.
            \item By definition,
            \[
                \winrA{\alpha^*}{X} \cap \winrD{\alpha^*}{Y} = \begin{aligned}[t]
                    &\bigcap \{Z_1\ \subseteq \s \;|\; X \cup \winrA{\alpha}{Z_1} \subseteq Z_1\}\\
                    & \cap \bigcup\{Z_2 \in \s \;|\; Z_2 \subseteq Y \cap \winrD{\alpha}{Z_2}\}
            \end{aligned} \]
            The least fixpoint in Demon's winning region can be upper bounded by a greatest fixpoint formed by all pre-fixpoints fulfilling $Y\cap \winrD{\alpha}{Z_2} \subseteq Z_2$. Intuitively, intersecting all sets that are supersets of $Y\cap \winrD{\alpha}{Z_2}$ yields a bigger set than unifying all sets that are subsets of $Y\cap \winrD{\alpha}{Z_2}$. As a result, $\bigcap \{Z_1\ \subseteq \s \;|\; X \cup \winrA{\alpha}{Z_1} \subseteq Z_1\}\cup \bigcap\{Z_2 \in \s \;|\; Y \cap \winrD{\alpha}{Z_2}\subseteq Z_2\}$ is obtained. $X\cup \winrA{\alpha}{Z_1}$ being a subset of $Z_1$ and $Y\cap \winrD{\alpha}{Y}$ being a subset of $Z_2$, implies that their intersection lies in $Z_1 \cap Z_2$, i.e. the set can be increased to $\bigcap\{Z_1, Z_2 \in \s \;|\; (X\cup \winrA{\alpha}{Z_1}) \cap Y \cap \winrD{\alpha}{Z_2} \subseteq Z_1 \cap Z_2\}$. Expanding the inner term yields $(X\cap Y \cap \winrD{\alpha}{Z_2}) \cup (\winrA{\alpha}{Z_1} \cap Y \cap \winrD{\alpha}{Z_2})$. Increasing this set by dropping $\winrD{\alpha}{Z_2}$ and $Y$ resp. from the left- and right-hand side of the union resp. (thus increasing the whole fixpoint) makes the IH applicable which yields $\{Z_1, Z_2 \in \s \;|\; (X\cap Y) \cup \winrA{\alpha^{-d}}{Z_1 \cap Z_2} \subseteq Z_1 \cap Z_2\}$. Renaming $Z_1 \cap Z_2$ to $Z$ returns $\bigcap\{Z\in \s \;|\; (X\cap Y) \cup \winrA{\alpha^{-d}}{Z} \subseteq Z\} = \winrA{(\alpha^*)^{-d}}{X\cap Y}$.
        \end{enumerate}
        \qed
    \end{proof}
    \begin{lemma}
    \label{incl_Angel}
        It holds that
        \[\winrA{\alpha}{\winrA{\beta^{-d}}{X \cap Y}} \cup \winrA{\beta}{X} \cap (\winrA{\alpha}{\winrA{\beta}{X}})^\C \subseteq \winrA{\alpha^{-d}}{\winrA{\beta^{-d}}{X \cap Y}}\]
    \end{lemma}
    \begin{proof}
        $\winrA{\alpha}{\winrA{\beta^{-d}}{X\cap Y} \cup \winrA{\beta}{X}} \cap (\winrA{\alpha}{\winrA{\beta}{X}})^\C = \winrA{\alpha}{\winrA{\beta^{-d}}{X\cap Y}\cup \winrA{\beta}{X}} \cap \winrD{\alpha}{\winrD{\beta}{X^\C}}$ by applying determinacy twice. Then, L. \ref{dgl cut rewrite} is used to merge the winning regions and increase them to $\winrA{\alpha^-d}{(\winrA{\beta^{-d}}{(X\cap Y)} \cup \winrA{\beta}{X}) \cap \winrD{\beta}{X^\C}}$. Expanding the term yields $\winrA{\alpha^{-d}}{(\winrA{\beta^{-d}}{X\cap Y} \cap \winrD{\beta}{X^\C}) \cup (\winrA{\beta} \cap \winrD{\beta}{X^\C})} \subseteq \winrA{\alpha^{-d}}{(\winrA{\beta^{-d}}{X\cap Y} \cap \winrD{\beta}{X^\C}) \cup \winrA{\beta^{-d}}{X \cap X^\C}}$ by L. \ref{dgl cut rewrite} and monotonicity. As $X \cap X^\C$ is empty, $\winrA{\beta^{-d}}{X\cap X^\C}$ must also be empty because the goal is unreachable and all games are under Angel's control, so she cannot win by Demon losing a test or being unable to evolve an ODE. Finally, $\winrA{\alpha^{-d}}{\winrA{\beta^{-d}}{X\cap Y} \cap \winrD{\beta}{X^\C}} \subseteq \winrA{\alpha^{-d}}{\winrA{\beta^{-d}}{X\cap Y}}$ by monotonicity.
        \qed
    \end{proof}
    \begin{lemma}
    \label{incl_Demon}
        It holds that:
        \[
        \winrD{\alpha}{\winrA{\beta^{-d}}{X \cap Y} \cup \winrD{\beta}{Y}} \cap (\winrD{\alpha}{\winrD{\beta}{Y}})^\C \subseteq \winrA{\alpha^{-d}}{\winrA{\beta^{-d}}{X \cap Y}}
        \]
    \end{lemma}
    \begin{proof}
        This proof works exactly the same as the one for L. \ref{incl_Angel} and is left as an exercise to the reader.
    \end{proof}
    With all those preliminaries proven, the actual proof can be conducted. Both equalities are shown by structural induction over the complexity of the hybrid game $\alpha$.
    \begin{enumerate}
        \item For the assignment game, test game for both Angel and Demon and the continuous game for Angel, $\winrA{\alpha}{X}{Y} = \winrA{\alpha^{-d}}{X\cap Y} \cup \winrA{\alpha}{X}$ and $\winrD{\alpha}{X}{Y} = \winrA{\alpha^{-d}}{X\cap Y}\cup \winrD{\alpha}{Y}$ can directly be proven by merging the definitions of the right-hand side winning regions and using $X\cup (X\cap Y) = X$ and $Y\cup(X\cap Y) = Y$ respectively.
        \item By definition,
        \[\winrD{x'=f(x)\& Q}{X}{Y} = \begin{aligned}[t]
            &\{\varphi(0) \in \s \;|\; \varphi(r) \in Y \text{ for all } r \text{ with } \varphi \models x'=f(x) \land Q\}\\
            & \cup \begin{aligned}[t]
                \{\varphi(0) \in \s \;|\; \varphi(r) \in X\cap Y \text{ for some } r \hspace{2.5cm}&\\
                \text{ with } \varphi \models x'=f(x) \land Q\}&
            \end{aligned}
        \end{aligned}
        \] 
        which directly corresponds to $\winrA{(x'=f(x) \& Q)^{-d}}{X\cap Y} \cup \winrD{x'=f(x) \& Q}{Y}$.
        \item $\winrA{(?Q)^{-d}}{X\cap Y} \cup \winrD{?Q}{Y} = (\winr{Q} \cap (X \cap Y)) \cup (\winr{Q}^\C \cup Y)$ by definition. Expanding the term yields $(\winr{Q} \cup \winr{Q}^\C \cup Y) \cap ((X\cap Y) \cup \winr{Q}^\C \cup Y)$. As the left-hand side of the intersection is the set of all states $\s$, the term reduces to the right-hand side. As $(X\cap Y) \cup Y = Y$, this can be transformed to $\winr{Q}^C \cup Y = \winrD{?Q}{X}{Y}$.
        \item By definition, $\winrA{\alpha\cup\beta}{X}{Y} = \winrA{\alpha}{X}{Y} \cup \winrA{\beta}{X}{Y}$. Using the IH and reordering yields $(\winrA{\alpha^{-d}}{X\cap Y} \cup \winrA{\beta^{-d}}{X\cap Y}) \cup (\winrA{\alpha}{X} \cup \winrA{\beta}{X}) = \winrA{(\alpha\cup\beta)^{-d}}{X\cap Y} \cup \winrA{\alpha\cup\beta}{X}$.
        
        By definition, \[\winrD{\alpha\cup\beta}{X}{Y} = \begin{aligned}[t]
            &(\winrD{\alpha}{X}{Y} \cap \winrD{\beta}{X}{Y})\\
            &\cup (\winrD{\alpha}{X}{Y} \cap \winrA{\alpha}{X}{Y})\\
            &\cup (\winrD{\beta}{X}{Y} \cap \winrA{\beta}{X}{Y})
        \end{aligned}\]
        Then, the IH is applied to all parts, yielding
        \[\begin{aligned}[t]
            &((\winrA{\alpha^{-d}}{X\cap Y} \cup \winrD{\alpha}{Y}) \cap (\winrA{\beta^{-d}}{X\cap Y} \cup \winrD{\beta}{Y}))\\
            &\cup ((\winrA{\alpha^{-d}}{X\cap Y} \cup \winrD{\alpha}{Y}) \cap (\winrA{\alpha^{-d}}{X\cap Y} \cup \winrA{\alpha}{X}))\\
            &\cup ((\winrA{\beta^{-d}}{X\cap Y} \cup \winrD{\beta}{Y}) \cap (\winrA{\beta^{-d}}{X\cap Y} \cup \winrA{\beta}{X}))
        \end{aligned}\]
        $(\winrA{\alpha^{-d}}{X\cap Y} \cup \winrD{\alpha}{Y}) \cap (\winrA{\alpha^{-d}}{X\cap Y} \cup \winrA{\alpha}{X}) = (\winrA{\alpha^{-d}}{X\cap Y} \cup (\winrA{\alpha}{X} \cap \winrD{\alpha}{Y}))$ and $(\winrA{\beta^{-d}}{X\cap Y} \cup \winrD{\beta}{Y}) \cap (\winrA{\beta^{-d}}{X\cap Y} \cup \winrA{\beta}{X}) = \winrA{\beta^{-d}} \cup (\winrA{\beta}{X}\cap \winrD{\beta}{Y})$, so with L. \ref{dgl cut rewrite} this collapses to $\winrA{\alpha^{-d}}{X\cap Y}$ and $\winrA{\beta^{-d}}{X\cap Y}$ respectively. Expanding and simplifying the term yields $\winrA{\alpha^{-d}}{X\cap Y} \cup \winrA{\beta^{-d}}{X\cap Y} \cup (\winrD{\alpha}{Y} \cap \winrD{\beta}{Y}) = \winrA{(\alpha\cup \beta)^{-d}}{X\cap Y} \cup \winrD{\alpha\cup\beta}{Y}$.
        \item $\winrA{\alpha;\beta}{X}{Y} = \winrA{\alpha}{\winrA{\beta}{X}{Y}}{\winrD{\beta}{X}{Y}}$ by definition. Using the IH twice yields $\winrA{\alpha^{-d}}{\winrA{\beta^{-d}}{X\cap Y} \cup (\winrA{\beta}{X} \cap \winrD{\beta}{Y})} \cup \winrA{\alpha}{\winrA{\beta^{-d}}{X\cap Y} \cup \winrA{\beta}{X}}$. The left-hand side of the union is equal to $\winrA{\alpha^{-d}}{\winrA{\beta^{-d}}{X\cap Y}}$ according to L. \ref{dgl cut rewrite}. The left-hand side of the union is intersected with $\s$ in the form of $\winrA{\alpha}{\winrA{\beta}{X}} \cup \winrA{\alpha}{\winrA{\beta}{X}}^\C$. By expanding, the union of $\winrA{\alpha}{\winrA{\beta^{-d}}{X\cap Y} \cup \winrA{\beta}{X}} \cap \winrA{\alpha}{\winrA{\beta}{X}} = \winrA{\alpha}{\winrA{\beta}{X}}$ by monotonictiy, and $\winrA{\alpha}{\winrA{\beta^{-d}}{X\cap Y} \cup \winrA{\beta}{X}} \cap  \winrA{\alpha}{\winrA{\beta}{X}}^\C \subseteq \winrA{\alpha^{-d}}{\winrA{\beta^{-d}}{X\cap Y}}$ by L. \ref{incl_Angel}. Consequently, the whole term can be transformed to $\winrA{\alpha^{-d}}{\winrA{\beta^{-d}}{X\cap Y}} \cup \winrA{\alpha}{\winrA{\beta}{X}} = \winrA{(\alpha;\beta)^{-d}}{X\cap Y} \cup \winrA{\alpha;\beta}{X}$. Similarly, the claim can be proven for Demon, only using L. \ref{incl_Demon} instead of L. \ref{incl_Angel}.
        \item $\winrA{\alpha^d}{X}{Y} = \winrD{\alpha}{Y}{X}$ by definition. Using the IH yields $\winrA{\alpha^{-d}}{X\cap Y} \cup \winrD{\alpha}{X} = \winrA{(\alpha^d)^{-d}}{X\cap Y} \cup \winrA{\alpha^d}{X}$ by definition of $^{-d}$ and determinacy, definition of winning region for the dual operator. The proof for Demon works similarly.
        \item By definition, 
        \[\winrA{\alpha^*}{X}{Y} = \begin{aligned}[t]
            &\bigcap\{Z \subseteq \s \;|\; X \cup \winrA{\alpha}{Z}{Z^\C} \subseteq Z\}\\
            &\cup \bigcap\{Z \subseteq \s \;|\; (X\cap Y) \cup (\winrA{\alpha}{Z}{Z} \cap \winrD{\alpha}{Z}{Z}) \subseteq Z\}
        \end{aligned}\]
        Applying the IH to $\winrA{\alpha}{Z}{Z} \cap \winrD{\alpha}{Z}{Z}$ yields $(\winrA{\alpha^{-d}}{Z} \cup \winrA{\alpha}{Z}) \cap (\winrA{\alpha^{-d}}{Z} \cup \winrD{\alpha}{Z}) = \winrA{\alpha^{-d}}{Z} \cup (\winrA{\alpha}{Z} \cap \winrD{\alpha}{Z})$. According to L. \ref{dgl cut rewrite}, this collapses to $\winrA{\alpha^{-d}}{Z}$. Additionally, $\winrA{\alpha}{Z}{Z^\C} = \winrA{\alpha}{Z}$ by L. \ref{CollapseDGL}. Consequently, the whole term can be transformed to $\bigcap\{Z\subseteq \s\;|\; X \cup \winrA{\alpha}{Z} \subseteq Z\} \cup \bigcap\{Z \subseteq \s \;|\; (X\cap Y) \cup \winrA{\alpha^{-d}}{Z} \subseteq Z\} = \winrA{(\alpha^*)^{-d}}{X\cap Y} \cup \winrA{\alpha^*}{X}$. The proof for Demon is conducted similarly.
    \end{enumerate}
    \qed
\end{proofE}

As a corollary for Lemma \ref{rewriting}, it can be stated that players will cooperate, iff both can win the game.
\begin{corollaryE}[Joint Cooperation][end]
\label{A+D}
    It holds that
    \[
    \winrA{\alpha}{X}{Y} \cap \winrD{\alpha}{X}{Y} = \winrA{\alpha^{-d}}{X\cap Y}
    \]
\end{corollaryE}
\begin{proofE}
    By L. \ref{rewriting},
    \[\winrA{\alpha}{X}{Y} \cap \winrD{\alpha}{X}{Y} = (\winrA{\alpha^{-d}}{X \cap Y} \cup \winrA{\alpha}{X}) \cap (\winrA{\alpha^{-d}}{X\cap Y} \cup \winrD{\alpha}{Y})\]
    This can then be transformed to $\winrA{\alpha^{-d}}{X \cap Y} \cup (\winrA{\alpha}{X} \cap \winrD{\alpha}{Y}) = \winrA{\alpha^{-d}}{X\cap Y}$ by L. \ref{dgl cut rewrite}.
    \qed
\end{proofE}

\begin{lemmaE}[Goals to tests conversion][all end]
    \label{test_conversion}
    Goals can be converted into tests:
    \[{\small
        \angel{\alpha}{P}{Q} \leftrightarrow \angel{\alpha; ?P ;(?Q)^d}{\top}{\top} \land \angel{\alpha;(?Q)^d; ?P}{\top}{\top}} 
    \]
    \[{\small
        \demon{\alpha}{P}{Q} \leftrightarrow \demon{\alpha; ?P ;(?Q)^d}{\top}{\top} \land \demon{\alpha;(?Q)^d; ?P}{\top}{\top} }
    \]
To make sure that Angel, resp. Demon not only wins because the other one fails the test, they have to be able to win the game for any sequence of the tests.
\end{lemmaE}
\begin{proofE}
    The claim is shown for Angel and Demon by proving that the winning regions of both sides of the equivalence are equal.
    \begin{multline*}
    \winr{ \angel{\alpha; ?P ;(?Q)^d}{\top}{\top} \land \angel{\alpha;(?Q)^d; ?P}{\top}{\top}} =\\ \winrA{\alpha;?P;(?Q)^d}{\s}{\s} \cap \winrA{\alpha;(?Q)^d;?P}{\s}{\s}
    \end{multline*}
    by definition. Simplifying the term as much as possible using the winning regions' definitions yields $\winrA{\alpha}{\winr{P}}{\winr{Q} \cup \winr{P}^\C} \cap \winrA{\alpha}{\winr{Q}^\C\cup \winr{P}}{\winr{Q}}$. Using L. \ref{rewriting} on both parts of the intersection transform the term to 
    \[(\winrA{\alpha^{-d}}{\winr{P}\cap (\winr{Q} \cup \winr{P}^\C)} \cup \winrA{\alpha}{\winr{P}}) \cap (\winrA{\alpha^{-d}}{(\winr{Q}^\C \cup \winr{P})\cap \winr{Q}} \cup \winrA{\alpha}{\winr{Q}^\C}\cup \winr{P})\]
    As both $\winr{P}\cap (\winr{Q} \cup \winr{P}^\C)$ and $(\winr{Q}^\C \cup \winr{P})\cap \winr{Q}$ are equivalent to $\winr{P} \cap \winr{Q}$, the term can be simplified to $\winrA{\alpha^{-d}}{\winr{P} \cap \winr{Q}} \cup (\winrA{\alpha}{\winr{P}} \cap \winrA{\alpha}{\winr{Q}^\C \cup \winr{P}})$. Due to L. \ref{monotonicity}, $\winrA{\alpha}{\winr{P}} \subseteq \winrA{\alpha}{\winr{Q}^\C \cup \winr{P}}$. Consequently, the term is equal to $\winrA{\alpha^{-d}}{\winr{P}\cap\winr{Q}} \cup \winrA{\alpha}{\winr{P}} = \winrA{\alpha}{\winr{P}}{\winr{Q}} = \winr{\angel{\alpha}{P}{Q}}$ by L. \ref{rewriting}. The proof for Demon works similarly.
    \qed
\end{proofE}

\section{Proof Calculus}
\label{proof_calculus}
For making practical use of \dGLsc, a proof calculus is provided in this section. The proof calculus includes two axioms for each hybrid game, one for Angel and one for Demon. They correspond to the game's semantic definitions and decompose formulas into smaller parts.

\begin{figure}[t]
    \centering
    \begin{tabularx}{\textwidth}{>{\color{gray}}l X}
        $\asA$ & $\angel{x:=e}{p(x)}{q(x)} \leftrightarrow p(e)$ \\
        $\conA$ & $\angel{x'=f(x)}{P}{Q} \leftrightarrow \exists t{\geq} 0\, \angel{x:= y(t)}{P}{Q}$ \hfill $(y' = f(y))$ \\
        $\tA$ & $\angel{?R}{P}{Q} \leftrightarrow R \land P$ \\
        $\chA$ & $\angel{\alpha\cup\beta}{P}{Q} \leftrightarrow \angel{\alpha}{P}{Q} \lor \angel{\beta}{P}{Q}$ \\
        $\seqA$ & $\angel{\alpha; \beta}{P}{Q} \leftrightarrow \angel{\alpha}{\angel{\beta}{P}{Q}}{\demon{\beta}{P}{Q}}$ \\
        $\dA$ & $\angel{\alpha^d}{P}{Q} \leftrightarrow \demon{\alpha}{Q}{P}$ \\
        $\rA$ & $\angel{\alpha^*}{P}{Q} \leftrightarrow P \lor \angel{\alpha; \alpha^*}{P}{Q}$ \\
        $\FP$ & \AxiomC{$P\lor \angel{\alpha}{R_1}\rightarrow R_1$}
        \AxiomC{$(P\land Q) \lor (\angel{\alpha}{R_2}{R_2} \land \demon{\alpha}{R_2}{R_2})   \rightarrow R_2$}
        \BinaryInfC{$\angel{\alpha^*}{P}{Q} \rightarrow R_1 \lor R_2$}
        \DisplayProof \\
        $\asD$ & $\demon{x:=e}{p(x)}{q(x)} \leftrightarrow q(e)$ \\
        $\conD$ & $\demon{x'=f(x)}{P}{Q} \leftrightarrow \forall t{\geq} 0\,\demon{x{:}{=} x(t)}{P}{Q} \lor \exists t{\geq} 0\,\angel{x{:}{=} x(t)}{P \land Q}{Q}$ \hfill$(x' = f(x))$ \\
        $\tD$ & $\demon{?R}{P}{Q} \leftrightarrow \lnot R \lor Q$ \\
        $\chD$ & $\demon{\alpha \cup \beta}{P}{Q}\leftrightarrow \big(\demon{\alpha}{P}{Q} \land \demon{\beta}{P}{Q}\big) \lor \big(\demon{\alpha}{P}{Q} \land \angel{\alpha}{P}{Q}\big) \lor \big(\demon{\beta}{P}{Q} \land \angel{\beta}{P}{Q}\big)$ \\
        $\seqD$ & $\demon{\alpha;\beta}{P}{Q} \leftrightarrow \demon{\alpha}{\angel{\beta}{P}{Q}}{\demon{\beta}{P}{Q}}$ \\
        $\rD$ & $\demon{\alpha^*}{P}{Q} \leftrightarrow (P \land Q) \lor (\angel{\alpha;\alpha^*}{P}{Q} \land \demon{\alpha;\alpha^*}{P}{Q}) \lor (Q \land \demon{\alpha; \alpha^*}{Q})$ \\
        $\MA$ & \AxiomC{$P_1 \rightarrow P_2$}
        \AxiomC{$Q_1 \rightarrow Q_2$}
        \BinaryInfC{$\angel{\alpha}{P_1}{Q_1} \rightarrow \angel{\alpha}{P_2}{Q_2}$}
        \DisplayProof\\
        $\MAtwo$ & \AxiomC{$P_1 \rightarrow P_2$}
        \AxiomC{$P_1 \land Q_1 \rightarrow \bot$}
        \BinaryInfC{$\angel{\alpha}{P_1}{Q_1} \rightarrow \angel{\alpha}{P_2}{Q_2}$}
        \DisplayProof\\
        $\detR$ & $\lnot \angel{\alpha}{P} \leftrightarrow \demon{\alpha}{\lnot P}$
    \end{tabularx}
    \caption{Proof calculus for \dGLsc}
    \label{dglA_calculus}
\end{figure}

Additionally, there is the determinacy axiom $\det$, which corresponds to Corollary \ref{determinancy} and links Angel and Demon's modality, if they have opposing goals.
The fixpoint rule \FP{} is, like the other rules, written with the premises at the top and the conclusion at the bottom. The rule characterizes $\angel{\alpha^*}{P}{Q}$ as the union of two least fixpoints. If $R_1$ is a pre-fixpoint for the first fixpoint and $R_2$ is a pre-fixpoint for the second, then $R_1\lor R_2$ holds whenever $\angel{\alpha^*}{P}{Q}$ holds. As a subtle detail, $R_1$ and $R_2$ do not need to be the same because Angel can rarely achieve the same by competing and by cooperating with Demon. Unlike axiom $\rA$, the rule does not feature more instances of repetition games in the premises. Demon's induction rule is not included because it can be derived from \FP{}.
The other two proof rules $\MA$ and $\MAtwo$ correspond to the two monotonicity lemmas. Unlike the axioms, the monotonicity rules for Demon can be derived by using $\dA$. For the monotonicity rules, it is irrelevant that this changes the game. But for Demon's axioms, the dual game introduced by $\dA$ prevents the use of any other axiom to derive them.

Additionally, the proof calculus includes all rules for first-order logic. The proof calculus' rules are listed in Fig. \ref{dglA_calculus}, except for the first-order rules which are left out for reasons of space limitations and readability.

\begin{example}
    With these proof rules on hand, it can now be proved that Angel is actually able to win the game seen in Example \ref{kayak} and gets enough orange juice to supply all her passengers but not so much that it spills everywhere.
    In the proof, the hybrid game is straightforwardly broken down using the continuous, dual and assignment axioms:
    \begin{small}
        \begin{prooftree}
        \def\fCenter{\ \vdash\ }
        \AxiomC{$*$}
        \LeftLabel{$\mathbb{R}$}
        \UnaryInf $o=0, t=0 \fCenter y+5 = 5 \land o+3=3$
        \LeftLabel{$\lor$R, wR, $\asA$, $\exists$R}
        \UnaryInf $o=0, t=0 \fCenter \begin{aligned}[t]
        \forall s{\geq}0 \demon{o:=o+s}{t+5 = 5}{o=3}\hspace{1.9cm}&\\
        \lor \exists s {\geq} 0 \angel{o=o+s}{t+5=5 \land o=3}{o=3}&
        \end{aligned}$
        \LeftLabel{$\conD$}
        \UnaryInf $o=0, t=0 \fCenter \demon{x'=1}{t+5 = 5}{o=3}$
        \LeftLabel{$\exists$R, $\asA$}
        \UnaryInf $o=0, t=0 \fCenter\begin{aligned}[t]
         \exists s{\geq}0 \angel{t:=t+s}{\demon{x'=1}{t=5}{o=3}}{\hspace{1.2cm}&\\
        \angel{x'=1}{t=5}{o=3}}&
        \end{aligned}$
        \LeftLabel{$\conA$}
        \UnaryInf $o=0; t=0 \fCenter \begin{aligned}[t]
        \angel{y'=1}{\demon{x'=1}{t=5}{o=3}}{\hspace{2.4cm}&\\
        \angel{x'=1}{t=5}{o=3}}&
        \end{aligned}$
        \LeftLabel{$\seqA$, $\dA$, $\dD$}
        \UnaryInf $o=0, t=0 \fCenter \angel{\{t'=1\}; \{o'=1\}^d}{o=3}{t=5}$
    \end{prooftree}
    \end{small}
    Notably, the point where the proof splits into two cases (cooperation or competition), only occurs in the penultimate step of the proof. In this case it is already conceivable which option is provable, and the other branch can be eliminated, but usually this is hard to see. 
    Therefore, this saves considerable effort compared to using \dGL where both cases had to be considered from the start.   
\end{example}

To be of any use, soundness of the proof calculus is crucial. Soundness means that all formulas that can be proven in the proof calculus are actually valid. If a false formula could be proven, anything could be, rendering it useless. Therefore, the soundness of the proof calculus is proven in the following theorem:
\begin{theoremE}[Soundness][end]
\label{soundness}
    The \dGLsc proof calculus is sound.
\end{theoremE}
\begin{proofE}
    The proof is done by showing that the semantics of the left-hand side and the right-hand side of all axioms match. The proof rules are shown by proving that the semantics of the premise is a subset of the semantics of the conclusion.
    \begin{enumerate}
        \item Axioms $\asA$, $\asD$:
        
            $\winr{\angel{x:=e}{p(x)}{q(x)}} = \winrA{x:=e}{\winr{p(x)}}{\winr{q(x)}} = \{\omega \in \s \;|; \omega^{\omega\winr{x}}_x \in \winr{p(x)}\}$ by definition. Instead of replacing the value of $x$ with the value of $e$ in $\omega$ and checking if the resulting state is in $\winr{p(x)}$, it can also be checked if $\omega$ is in $\winr{p(e)}$. Consequently, $\winr{\angel{x:=e}{p(x)}{q(x)}} = \winr{p(e)}$.

            The proof for $\asD$ works similarly.
        \item Axioms $\conA$, $\conD$:
        
            By definition $\winr{\exists t{\geq}0 \angel{x:=x(t)}{P}{Q}} = \{\omega\in \s \;|\; \omega^r_t \in \{\nu \in \s \;|\; \nu^{\nu\winr{x(t)}}_x \in \winr{P}\} \text{for some } r\geq 0\}$ by definition. Contracting the replacements made in $\omega$ and $\nu$ yields $\{\omega \in \s \;|\; (\omega^r_t)^{\omega^r_t\winr{x(t)}}_x \in \winr{P} \text{ for some } r\geq 0\}$. Define now the function of states $\varphi(s) = (\omega^s_t)^{\omega^r_t\winr{x(t)}}_x$ which contains all states that are equal to $\omega$ except for $t, x, x'$ and $x$ has the value $x(s)$ in state $\varphi(s)$. The condition $\varphi \models x'=f(x)$ is fulfilled because $x'=f(x)$ by assumption. Therefore, the set is a subset of $\winr{\angel{x'=f(x)}{P}{Q}} = \{\varphi(0) \in \s \;|\; \varphi(r) \in \winr{P} \text{ for some } r \text{ with } \varphi \models x'=f(x)\}$. The other inclusion holds because the solution of $s'=f(x)$ is unique as proven in \cite{Platzer10}.

            For $\conD$, 
            \[\winr{\demon{x'=f(x)}{P}{Q}} = \begin{aligned}[t]
                &\{\varphi(0) \in \s \;|\; \varphi(r) \in \winr{Q} \text{ for all } r \text{ with } \varphi \models x'=f(x)\}\\
                &\cup \{\varphi(0) \in \s \;|\; \varphi(r) \in \winr{P} \cap \winr{Q}\}
            \end{aligned} \] 
            and 
            \[\begin{aligned}[t]
            &\winr{\forall t{\geq}0 \demon{x:=x(t)}{P}{Q} \lor \exists t{\geq}0 \angel{x:=x(t)}{P\land Q}{Q}}\\
            &= \begin{aligned}[t]
                &\{\omega \in \s \;|\; (\omega^r_t)^{\omega^r_t\winr{x(t)}}_x \in \winr{Q} \text{ for all } r\geq 0\}\\
                &\cup \{\omega\in\s\;|\; (\omega^r_t)^{\omega^r_t\winr{x(t)}}_x \in \winr{P} \cap \winr{Q} \text{ for some } r \geq 0\}
            \end{aligned}
            \end{aligned}
            \] 
            It can now be proved that $\{\varphi(0) \in \s \;|\; \varphi(r) \in \winr{Q} \text{ for all } r \text{ with } \varphi \models x'=f(x)\}$ is equal to $\{\omega \in \s \;|\; (\omega^r_t)^{\omega^r_t\winr{x(t)}}_x \in \winr{Q} \text{ for all } r\geq 0\}$ and that $\{\varphi(0) \in \s \;|\; \varphi(r) \in \winr{P} \cap \winr{Q}\}$ is equal to $\{\omega\in\s\;|\; (\omega^r_t)^{\omega^r_t\winr{x(t)}}_x \in \winr{P} \cap \winr{Q} \text{ for some } r \geq 0\}$, similarly to the proof for $\conA$.
        \item 
            Axioms $\tA$, $\tD$, $\chA$, $\chD$, $\seqA$, $\seqD$, $\dA$,     $\dD$:

            These axioms can directly be proven by using the definitions of the semantics.
        \item Axiom $\rA$, $\rD$:
        
            $\winr{\angel{\alpha^*}}{P}{Q} = \winrA{\alpha^*}{\winr{P}}{\winr{Q}} = \winrA{\alpha^*}{\winr{P}} \cup \winrA{(\alpha^*)^{-d}}{\winr{P} \cap \winr{Q}}$ by L. \ref{rewriting}. As the goals are complementary for both winning regions, each of them consists of only one fixpoint because the other one is empty. Using the fixpoints' properties yields
            \[\winrA{\alpha^*}{\winr{P}} = \{Z \subseteq \s \;|\; \winr{P} \cup \winrA{\alpha}{Z, Z^\C} \subseteq Z\} = \winr{P} \cup \winrA{\alpha}{\winrA{\alpha^*}{\winr{P}}}\] 
            and \[\begin{aligned}[t]
                \winrA{(\alpha^*)^{-d}}{\winr{P} \cap \winr{Q}} &= \{Z \subseteq \s \;|\; (\winr{P} \cap \winr{Q}) \cup \winrA{\alpha^{-d}}{Z}{Z^C} \subseteq Z\}\\
                &= (\winr{P} \cap \winr{Q}) \cup \winrA{\alpha^{-d}}{\winrA{(\alpha^*)^{-d}}{\winr{P}\cap \winr{Q}}}
            \end{aligned}\] 
            Because of C. \ref{determinancy}, $\winrA{\alpha}{\winrA{\alpha^*}{\winr{P}}} = \winrA{\alpha;\alpha^*}{\winr{P}}$ and $\winrA{\alpha^{-d}}{\winrA{(\alpha^*)^{-d}}{\winr{P}\cap \winr{Q}}} = \winrA{(\alpha;\alpha^*)^{-d}}{\winr{P} \cap \winr{Q}}$. In total, this yields \[\winr{P} \cup \winrA{\alpha;\alpha^*}{\winr{P}} \cup (\winr{P}\cap \winr{Q}) \cup \winrA{(\alpha;\alpha^*)^{-d}}{\winr{P}\cap \winr{Q}}\]
            By L. \ref
            {rewriting}, $\winrA{\alpha;\alpha^*}{\winr{P}} \cup \winrA{(\alpha;\alpha^*)^{-d}}{\winr{P}\cap \winr{Q}} = \winrA{\alpha;\alpha^*}{\winr{P}}{\winr{Q}}$. $\winr{P} \cup (\winr{P}\cap \winr{Q})$ can be simplified to $\winr{P}$. Consequently, the whole term is equal to $\winr{P} \cup \winrA{\alpha;\alpha^*}{\winr{P}{\winr{Q}}} = \winr{P \lor \angel{\alpha;\alpha^*}{P}{Q}}$.

            The proof for $\rD$ at first follows the same steps as the one for $\rA$ to show that $\winr{\demon{\alpha^*}{P}{Q}} = (\winr{Q} \cap \winrD{\alpha;\alpha^*}{Q}) \cup (\winr{P} \cap \winr{Q}) \cup \winrA{(\alpha;\alpha^*)^{-d}}{\winr{P} \cap \winr{Q}}$. As $\winrA{(\alpha;\alpha^*)^{-d}}{\winr{P} \cap \winr{Q}} = \winrA{\alpha;\alpha^*}{\winr{P}}{\winr{Q}} \cap \winrD{\alpha;\alpha^*}{\winr{P}}{\winr{Q}}$ by C. \ref{A+D}, the whole term is equal to $(\winr{Q} \cap \winrD{\alpha;\alpha^*}{Q}) \cup (\winr{P} \cap \winr{Q}) \cup (\winrA{\alpha;\alpha^*}{\winr{P}}{\winr{Q}} \cap \winrD{\alpha;\alpha^*}{\winr{P}}{\winr{Q}}) = \winr{(P\land Q) \lor (\angel{\alpha;\alpha^*}{P}{Q} \land \demon{\alpha;\alpha^*}{P}{Q}) \lor (Q \land \demon{\alpha;\alpha}{Q})}$.
        \item Rule $\MA$:
        
            Assume that the two premises are valid. Consequently, $\winr{P_1} \subseteq \winr{P_2}$ and $\winr{Q_1} \subseteq \winr{Q_2}$. According to Lemma \ref{monotonicity}, this implies that $\winrA{\alpha}{\winr{P_1}}{\winr{Q_2}} \subseteq \winrA{\alpha}{\winr{P_2}}{\winr{Q_2}}$. Therefore, the conclusion $\angel{\alpha}{P_1}{Q_2} \rightarrow \angel{\alpha}{P_2}{Q_2}$ is valid.
        \item Rule $\MAtwo$:
        
            Assume that the two premises are valid. Then, $\winr{P_1} \subseteq \winr{P_2}$. Additionally, $\winr{P_1} \cap \winr{Q_1} \subseteq \emptyset$ which implies $\winr{P_1} \cap \winr{Q_1} = \emptyset$ because the other inclusion is trivial. Therefore, all preconditions of Lemma \ref{mon2} are fulfilled. According to Lemma \ref{mon2}, this implies that $\winrA{\alpha}{\winr{P_1}}{\winr{Q_1}} \subseteq \winrA{\alpha}{\winr{P_2}}{\winr{Q_2}}$ holds. Consequently, the conclusion of the proof rule holds which proves that the proof rule is valid.
        \item Rule $\detR$:
        
            This rule is a direct consequence of Corollary \ref{determinancy}.
    \end{enumerate}
    \qed
\end{proofE}

Another important property of a proof calculus is completeness. Completeness means that any valid formula can be proven in the proof calculus. But how can this property be shown for \dGLsc? For \dGL it is already known that it is relatively complete \cite{Platzer15}. This means that any valid formula can be proved if all tautologies may be assumed in the antecedent, i.e. the notion is less strict than completeness. For \dGLsc only relative completeness can be proven, as this logic builds on top of \dGL. Adapting the definitions from \cite{Platzer15}, relative completeness is defined as:

\begin{definition}[Expressiveness \cite{Platzer15}]
    A logic $L$ is \emph{expressive} (for \dGLsc), if for every formula $\phi$, there exists an equivalent formula $\phi^\flat$ of $L$, i.e. $\models \phi \leftrightarrow \phi^\flat$.
    The logic $L$ is called \emph{differentially expressive}, if it is expressive and all equivalences of the form $\angel{x'=\theta}{P}{Q} \leftrightarrow (\angel{x'=\theta}{P}{Q})^\flat$ and $\demon{x'=\theta}{P}{Q} \leftrightarrow (\demon{x'=\theta}{P}{Q})^\flat$ are provable in the proof calculus. It is assumed that the logic $L$ is closed under first-order logic. 
\end{definition}

\begin{definition}[Relative completeness \cite{Platzer15}]
    A logic is \emph{complete relative to an expressive logic $L$}, if every valid formula can be proved in the calculus from $L$ tautologies.
\end{definition}

If now, the \dGLsc proof calculus could be converted to the \dGL proof calculus, relative completeness could be proven almost for free. In the semantics this conversion is possible by Lemma \ref{rewriting}. 
The proof calculus though, does not contain any axioms that syntactically rephrase this lemma.
The inclusion of a corresponding axiom would also run counter to the intended compositionality principles of \dGLsc.
Fortunately, this is also unnecessary, as Theorem \ref{Complementarization} shows a core insight: The two complementarization axioms (C.A. 1, 2) that split one modality with two individual player goals into two modalities with complementary goals, are \emph{admissible!} Consequently, the set of provable formulas is unaffected by adding them.
Hence, similar to Gentzen's cut elimination theorem, the compositional axioms of \dGLsc prove every \dGLsc formula that the complementarization axioms would reduce to \dGL, thereby guaranteeing relative completeness for free.

\begin{theoremE}[Complementarization Elimination][end]
    \label{Complementarization}
    The complementarization axioms are admissible in the proof calculus:
    \begin{itemize}
        \item Complementarization axiom 1: $\angel{\alpha}{P}{Q} \leftrightarrow \angel{\alpha^{-d}}{(P\land Q)} \lor \angel{\alpha}{P}$
        \item Complementarization axiom 2: $\demon{\alpha}{P}{Q} \leftrightarrow \angel{\alpha^{-d}}{(P\land Q)} \lor \demon{\alpha}{Q}$
    \end{itemize}  
\end{theoremE}
\begin{proofE}
    To facilitate the proof, some admissible or derivable rules are introduced first.
    \begin{lemma}
        \label{devRules}
        The following rules are admissible or derivable in the proof calculus:
        \begin{itemize}
            \item {\makebox[\cuswidth]{$\impAD$\hfill} $\angel{\alpha}{P} \land \demon{\alpha}{Q} \rightarrow \angel{\alpha^{-d}}{(P\land Q)}$}
            \item {\makebox[\cuswidth]{$\spA$\hfill} $\angel{\alpha}{(\angel{\beta^{-d}}{(P\land Q)}\lor \angel{\beta}{P})} \land \lnot \angel{\alpha}{\angel{\beta}{P}} \rightarrow \angel{\alpha^{-d}}{\angel{\beta^{-d}}{(P\land Q)}}$}
            \item {\makebox[\cuswidth]{$\id$\hfill} $\angel{\alpha^{-d}}{(P\land Q)} \leftrightarrow \angel{\alpha^{-d}}{(P\land Q)} \lor (\angel{\alpha}{P} \land \demon{\alpha}{Q})$}
            \item {\makebox[\cuswidth]{$\re$\hfill}$\begin{aligned}[t]
                &\angel{\alpha^{-d}}{\angel{\beta^{-d}}{(P\land Q)}} \lor \angel{\alpha}{(\angel{\beta^{-d}}{(P\land Q)} \lor \angel{\beta}{P})}\\[-3pt]
                &\leftrightarrow \angel{\alpha^{-d}}{\angel{\beta^{-d}}{(P\land Q)}} \lor \angel{\alpha}{\angel{\beta}{P}}
            \end{aligned}$}
            \item {\makebox[\cuswidth]{$\reD$\hfill} $\begin{aligned}[t]
                &\angel{\alpha^{-d}}{\angel{\beta^{-d}}{(P\land Q)}} \lor \demon{\alpha}{(\angel{\beta^{-d}}{(P\land Q)} \lor \demon{\beta}{Q})}\\[-3pt]
                &\leftrightarrow \angel{\alpha}{\angel{\beta^{-d}}{(P\land Q)}} \lor \demon{\alpha}{\demon{\beta}{Q}}
            \end{aligned}$}
            \item {\makebox[\cuswidth]{$\AD$\hfill} $\angel{\alpha}{P}{Q} \land \demon{\alpha}{P}{Q} \leftrightarrow \angel{\alpha^{-d}}{(P \land Q)}$}
            \item {\makebox[\cuswidth]{$\langle^*]$\hfill} $\angel{\alpha^*}{P}{Q}\land \demon{\alpha^*}{P}{Q} \leftrightarrow (P\land Q) \lor (\angel{\alpha;\alpha^*}{P}{Q} \land \demon{\alpha;\alpha^*}{P}{Q})$}
            \item \AxiomC{$(P\land Q) \lor (\angel{\alpha}{R}{R} \land \demon{\alpha}{R}{R}) \rightarrow R$}
            \LeftLabel{$\FPtwo$}
            \UnaryInfC{$(\angel{\alpha^*}{P}{Q} \land \demon{\alpha^*}{P}{Q}) \rightarrow R$}
            \DisplayProof
            \item \AxiomC{$Q\rightarrow \demon{\alpha}{Q}$}
            \LeftLabel{$\ind$}
            \UnaryInfC{$Q \rightarrow \demon{\alpha^*}{P}{Q}$}
            \DisplayProof
        \end{itemize}
    \end{lemma}
    \begin{proof}
        The proofs for these rules are left as an exercise for the reader.
    \end{proof}
    
    Both claims are proven simultaneously, using a structural induction over the complexity of the game $\alpha$.
    \begin{enumerate}
        \item $\angel{x:=e}{p(x)}{q(x)}$ holds iff $p(e)$ holds, by axiom $\asA$. Equivalently, $p(e)$ can be written as $(p(e) \land q(e)) \lor p(e)$. Using $\asA$ again, then yields $\angel{(x:=e)^{-d}}{p(x)\land q(x)} \lor \angel{x:=e}{p(x)}$.
        
        The proof for Demon's case works similarly by rewriting $q(e)$ to $(p(e) \land q(e)) \lor q(e)$.
        \item By $\conA$, $\angel{x'=f(x)}{p(x)}{q(x)}$ is equivalent to $\exists t{\geq}0 \angel{x:=x(t)}{p(x)}{q(x)}$. Applying $\asA$ removes all games, yielding $\exists t{\geq} 0 p(x(t))$. $p(x(t))$ can again be rewritten to $(p(x(t)) \land q(x(t))) \lor p(x(t))$. Additionally, the existential quantifier can be split in two, yielding $\exists t{\geq}0(p(x(t)) \land q(x(t))) \lor \exists t{\geq}0 p(x(t))$. Applying the axioms $\asA$ and $\conA$ backwards, transforms the formula to $\angel{(x'=f(x))^{-d}}{p(x)\land q(x)} \lor \angel{x'=f(x)}{p(x)}$.
        \item $\demon{x'=f(x)}{p(x)}{q(x)}$ is equivalent to $\forall t{\geq} 0\demon{x:=x(t)}{p(x)}{q(x)} \lor \exists t {\geq} 0 \angel{x:=x(t)}{p(x) \land q(x)}{q(x)}$. Afterward, the assignments are removed using $\asA$ and $\asD$. The result is then expanded by adding a disjunction with false, yielding $\forall t{\geq}0 q(x(t)) \lor \exists t{\geq}0 (q(t) \land \lnot q(t)) \lor \exists t {\geq} 0 (p(x) \land q(x))$. Finally, the formula is transformed to $\angel{(x'=f(x))^{-d}}{p(x)\land q(x)} \lor \demon{x'=f(x)}{q(x)}$ using the axioms $\asA$, $\asD$, $\conA$, $\conD$.
        \item By $\tA$, $\angel{?R}{P}{Q}$ is equivalent to $R \land P$. Rewriting $P$ to $(P\land Q) \lor P$, expanding the resulting formula and using $\tA$ again, yields $\angel{(?R)^{-d}}{P\land Q} \lor \angel{?R}{P}$.
        
        Similarly, $\angel{(?R)^{-d}}{P\land Q} \lor \demon{?R}{Q}$ is equivalent to $(R \land P \land Q) \lor \lnot R \lor Q$ by $\tA$, $\tD$. Expanding the formula yields $(R \lor \lnot R \lor Q) \land ((P \land Q) \lor \lnot R \lor Q)$. As the left-hand side of the conjunction is equivalent to \emph{true} and $(P\land Q) \lor Q$ is equivalent to $Q$, this can be simplified to $\lnot R \lor Q \leftrightarrow \demon{?R}{P}{Q}$ by $\tD$.
        \item $\angel{\alpha\cup \beta}{P}{Q} \leftrightarrow \angel{(\alpha\cup\beta)^{-d}}{(P\land Q)} \lor \angel{\alpha\cup \beta}P$ can directly be proven by first using $\chA$, then applying the IH to the result and using $\chA$ again.
        
        By $\chD$, \[\demon{\alpha\cup\beta}{P}{Q} \leftrightarrow \begin{aligned}[t]
            &(\demon{\alpha}{P}{Q} \land \demon{\beta}{P}{Q})\\
            &\lor (\demon{\alpha}{P}{Q} \land \angel{\alpha}{P}{Q})\\
            &\lor (\demon{\beta}{P}{Q} \land \angel{\beta}{P}{Q})
        \end{aligned}\]
        Then, the IH is applied to all parts, yielding:
        \[\begin{aligned}
            &((\angel{\alpha^{-d}}{(P\land Q)} \lor \demon{\alpha}{Q}) \land (\angel{\beta^{-d}}{(P\land Q)} \lor \demon{\beta}{Q}))\\
            &\lor ((\angel{\alpha^{-d}}{(P\land Q)} \lor \demon{\alpha}{Q}) \land (\angel{\alpha^{-d}} \lor \angel{\alpha}{P}))\\
            &\lor ((\angel{\beta^{-d}}{(P\land Q)} \lor \demon{\beta}{Q}) \land (\angel{\beta^{-d}}{(P\land Q)} \lor \angel{\beta}{P}))
        \end{aligned}\]
        $(\angel{\alpha^{-d}}{(P\land Q)} \lor \demon{\alpha}{Q}) \land (\angel{\alpha^{-d}} \lor \angel{\alpha}{P}) \leftrightarrow \angel{\alpha^{-d}}{(P\land Q)} \lor (\angel{\alpha}{P} \land \demon{\alpha}{Q})$ and $(\angel{\beta^{-d}}{(P\land Q)} \lor \demon{\beta}{Q}) \land (\angel{\beta^{-d}}{(P\land Q)} \lor \angel{\beta}{P}) \leftrightarrow \angel{\beta^{-d}}{(P\land Q)} \lor (\angel{\beta}{P} \land \demon{\beta}{Q})$, so with axiom $\id$, this can be transformed to $\angel{\alpha^{-d}}{(P\land Q)}$ and $\angel{\beta^{-d}}{(P\land Q)}$ respectively. Expanding and simplifying the formula yields $(\angel{\alpha^{-d}}{(P\land Q)} \lor \angel{\beta^{-d}}{(P\land Q)}) \lor (\demon{\alpha}{Q} \land \demon{\beta}{Q})$. Adding two disjunctions with \emph{false}, in the form of $(\demon{\alpha}{Q} \land \lnot \demon{\alpha}{Q}) \leftrightarrow (\demon{\alpha}{Q} \land \lnot \angel{\alpha}{\lnot Q})$ and $(\demon{\beta}{Q} \land \lnot \demon{\beta}{Q}) \leftrightarrow (\demon{\beta}{Q} \land \lnot \angel{\beta}{\lnot Q})$ by axiom $\det$, brings the formula in the right shape to use $\chA$, $\chD$, yielding $\angel{(\alpha \cup \beta)^{-d}}{(P\land Q)} \lor \demon{\alpha \cup \beta}{Q}$.
        \item $\angel{\alpha; \beta}{P}{Q} \leftrightarrow \angel{\alpha^{d}}{(\angel{\beta^{-d}}{(P\land Q)} \lor (\angel{\beta}{P} \land \demon{\beta}{Q}))} \lor \angel{\alpha}{(\angel{\beta^{-d}}{P \land Q} \lor \angel{\beta}{P})}$ by using $\seqA$ once and the IH twice. $\angel{\beta^{-d}}{(P\land Q)} \lor (\angel{\beta}{P} \land \demon{\beta}{Q})$ is equivalent to $\angel{\beta^{-d}}{(P\land Q)}$ by axiom $\id$. The result can then be transformed to $\angel{\beta^{-d}}{(P\land Q)} \lor \angel{\alpha}{\angel{\beta}{P}}$ using axiom $\re$. By $\det$, $\seqA$, this is equivalent to $\angel{(\alpha;\beta)^{-d}}{(P\land Q)} \lor \angel{\alpha;\beta}{P}$.
        
        The proof for Demon's case works similarly.
        \item For the dual game, the claim can be directly proven for both cases by using axiom $\dA$ and $\dD$ resp., then applying the IH and using $\dA$ and $dD$ resp., again.
        \item $\angel{\alpha^*}{P}{Q} \leftrightarrow P \lor \angel{\alpha;\alpha^*}{P}{Q}$ by $\rA$. Similarly to the proof for the sequential game, $\angel{\alpha;\alpha^*}{P}{Q}$ can be transformed to $\angel{(\alpha;\alpha^*)^{-d}}{(P\land Q)} \lor \angel{\alpha;\alpha^*}{P}$. $P$ is equivalent to $P \lor (P\land Q)$. Consequently, $\rA$ can be used again, yielding $\angel{(\alpha^{*})^{-d}}{(P\land Q)} \lor \angel{\alpha^*}{P}$.
        
        In Demon's case, $\demon{\alpha^*}{P}{Q} \leftrightarrow (P\land Q) \lor (\angel{\alpha;\alpha^*}{P}{Q} \land \demon{\alpha;\alpha^*}{P}{Q}) \lor (Q \land \demon{\alpha;\alpha^*}{Q})$ by $\rD$. By $\AD$, $\angel{\alpha;\alpha^*}{P}{Q} \land \demon{\alpha;\alpha^*}{P}{Q} \leftrightarrows \angel{(\alpha;\alpha^*)^{-d}}{(P\land Q)}$. Additionally, two conjunctions with \emph{false}, in the form of $Q \land \lnot Q$ and $\angel{(\alpha;\alpha^*)^{-d}}{Q \land  \lnot Q}$ are added. Now, the formula is in the right shape to use $\rA$, $\rD$, yielding $\angel{(\alpha^*)^{-d}}{(P\land Q)} \lor \demon{\alpha^*}{Q}$.
    \end{enumerate}
        \qed
\end{proofE}

Using Theorem \ref{completeness}, the relative completeness of \dGLsc can easily be proven:
\begin{theoremE}[Completeness][end, text link={\fullproofLocation}]
\label{completeness}
The \dGLsc proof calculus is complete relative to any differentially expressive logic $L$.
\end{theoremE}
\begin{proof}
    (Sketch) With complementarization, every formula of \dGLsc can be transformed to a formula where all modalities feature opposing goals. The complementarization axioms are admissible by Theorem~\ref{Complementarization}, so the rewrite can be mimicked using only the rules available in the \dGLsc proof calculus. For opposing goals, the rules in the proof calculus correspond to the rules of the proof calculus for \dGL which is complete relative to any differentially expressive logic $L$ \cite{Platzer15}.
    Consequently, the proof calculus for \dGLsc is also relatively complete.
\end{proof}
\begin{proofE}
    To show that the proof calculus is complete relative to any differentially expressive logic $L$, it has to be to proven that any valid formula $\varphi$ of \dGLsc can be derived from valid $L$ tautologies. Or, in signs: $\models \varphi$ implies $\vdash_L \varphi$. $\vdash_L \varphi$ means that a formula $\varphi$ can be derived from valid $L$ tautologies using the proof calculus.

    If $\varphi$ does not contain any hybrid games, then $\varphi$ is a first-order formula. Consequently, $\varphi$ is provable because it is assumed that the logic $L$ is closed under first-order connectives.

    If $\varphi$ contains hybrid games, all modalities are split in two, using the complementarization axioms. After that, $\varphi$ only contains modalities where Angel and Demon have opposing goals. As stated in Lemma \ref{CollapseDGL}, these modalities are equivalent to \dGL modalities. Additionally, the rules and axioms of the proof calculus correspond to the axioms and rules of the proof calculus for \dGL in the case that Angel and Demon have opposing goals. For the axioms, the transformation is straightforward:
    \begin{enumerate}
        \item $\angel{x:=e}{p(x)}{\lnot p(x)} \leftrightarrow p(e) \Leftrightarrow \angel{x:=e}{p(x)} \leftrightarrow p(e)$
        \item $\begin{aligned}[t]
            \angel{x'=f(x)}{P}{\lnot P} \leftrightarrow \exists{t{\geq}0}\angel{x:=x(t)}{P}{\lnot P}\hspace{4cm}\\
            \Leftrightarrow \angel{x'=f(x)}{P} \leftrightarrow \exists{t{\geq}0}\angel{x:=x(t)}{P}&
        \end{aligned}$
        \item $\angel{?R}{P}{\lnot P} \leftrightarrow R\land P\Leftrightarrow \angel{?R}{P} \leftrightarrow R\land P$
        \item $\angel{\alpha\cup\beta}{P}{\lnot P} \leftrightarrow \angel{\alpha}{P}{\lnot P} \lor \angel{\beta}{P}{\lnot P}\Leftrightarrow \angel{\alpha\cup\beta}{P} \leftrightarrow \angel{\alpha}{P} \lor \angel{\beta}{P}$
        \item $\begin{aligned}[t]
            \angel{\alpha;\beta}{P}{\lnot P} \leftrightarrow \angel{\alpha}{\angel{\beta}{P}{\lnot P}}{\demon{\beta}{P}{\lnot P}}\hspace{4.3cm}&\\
            \begin{aligned}[t]
                &\stackrel{\detR}{\Leftrightarrow} \angel{\alpha;\beta}{P}{\lnot P} \leftrightarrow \angel{\alpha}{\angel{\beta}{P}}{\lnot \angel{\beta}{P}}\\
                &\Leftrightarrow \angel{\alpha; \beta}{P}{\lnot P} \leftrightarrow \angel{\alpha}{\angel{\beta}{P}}
            \end{aligned}&
        \end{aligned}$
        \item $\angel{\alpha^d}{P}{\lnot P} \leftrightarrow \demon{\alpha}{\lnot P}{P} \Leftrightarrow \angel{\alpha^d}{P} \leftrightarrow \demon{\alpha}{P} \stackrel{\detR}{\Leftrightarrow} \angel{\alpha^d}{P} \leftrightarrow \lnot \angel{\alpha}{\lnot P}$
        \item $\angel{\alpha^*}{P}{\lnot P} \leftrightarrow P \lor \angel{\alpha;\alpha^*}{P}{\lnot P} \begin{aligned}[t]
            &\stackrel{\phantom{\seqA, \detR}}{\Leftrightarrow} \angel{\alpha^*}{P} \leftrightarrow P \lor \angel{\alpha;\alpha^*}{P}\\
            &\stackrel{\seqA, \detR}{\Leftrightarrow} \angel{\alpha^*}{P} \leftrightarrow P \lor \angel{\alpha}{\angel{\alpha^*}{P}}
        \end{aligned}$
        \item
        \hspace{5cm}
        \def\proofSkipAmount{\vskip -10pt}
        \begin{small}
            \begin{prooftree}
            \AxiomC{$P \lor \angel{\alpha}{(R)} \rightarrow R$}
                \AxiomC{$*$}
                \LeftLabel{$\bot L$}
                \UnaryInfC{$\bot \vdash \bot$}
                \LeftLabel{${\rightarrow} R$, prop. L.}
                \UnaryInfC{$\vdash \bot \land \bot \rightarrow \bot$}
                \LeftLabel{prop. L.}
                \UnaryInfC{$\vdash (P \land \lnot P) \lor (\angel{\alpha}{\bot}{\bot}  \land \demon{\alpha}{\bot}{\bot}) \rightarrow \bot$}
                \LeftLabel{\FP}
                \BinaryInfC{$\vdash \angel{\alpha^*}{P}{\lnot P} \rightarrow( R_1 \lor R_2) \lor \bot$}
                \LeftLabel{prop. L.}
                \UnaryInfC{$\vdash \angel{\alpha^*}{P}{\lnot P} \rightarrow R \lor \bot$}
        \end{prooftree}    
        \end{small}
        $\Leftrightarrow$
        \AxiomC{$P \lor \angel{\alpha}{R} \rightarrow R$}
        \LeftLabel{FP}
        \UnaryInfC{$\angel{\alpha}{P} \rightarrow R$}
        \DisplayProof

        To show the transformation for the fixpoint rule, $R_1 \lor R_2$, $R_2$ is instantiated with \emph{false} and $R_1$ is renamed to $R$. The right branch resulting from using the fixpoint rule, can be closed due to \emph{false} in the antecedent, while the other branch is the premise of the \dGL fixpoint rule.
        \item \detR:
        The rule \detR $\lnot \angel{\alpha}{P} \leftrightarrow \demon{\alpha}{\lnot P}$ is the same in both \dGLsc and \dGL.

        \item $\MAtwo$
        
        \AxiomC{$\vdash P \rightarrow Q$}
        \AxiomC{$*$}
        \LeftLabel{$\bot L$}
        \UnaryInfC{$\bot \vdash \bot$}
        \LeftLabel{$\rightarrow R$}
        \UnaryInfC{$\vdash P \land \lnot P \rightarrow \bot$}
        \LeftLabel{$\MAtwo$}
        \BinaryInfC{$\vdash \angel{\alpha}{P} \rightarrow \angel{\alpha}{Q}$}
        \DisplayProof
        $\Leftrightarrow$
        \AxiomC{$\vdash P \rightarrow Q$}
        \LeftLabel{$M\langle\rangle$}
        \UnaryInfC{$\vdash \angel{\alpha}{P} \rightarrow \angel{\alpha}{Q}$}
        \DisplayProof

        The second monotonicity rule can be collapsed to the monotonicity rule in \dGL, as the additional second branch can be trivially closed if Angel and Demon have opposing goals.
    \end{enumerate}
    
    Now, it can be seen that all proof rules for \dGL are included in the proof calculus of \dGLsc. It will therefore be omitted to transform the rest of the rules in the \dGLsc proof calculus because having more rules does not change the result and these rules can either be deduced from the rules above (all rules for Demon) or become trivial (the first monotonicity rule).

    In \cite{Platzer15}, Platzer already proves that \dGL is relatively complete to any differentially expressive logic $L$. As any \dGLsc formula can be transformed into a formula that is equivalent to a \dGL formula and the rules of the proof calculus are also equivalent to the \dGL rules, the exact same proof can be conducted. Consequently, the \dGLsc proof calculus is also relatively complete to any differentially expressive logic $L$.
    \qed
\end{proofE}

\section{Conclusion}

This paper introduces the logic \dGLsc for reasoning about hybrid games with two players that have individual goals. \dGLsc helps determine if a player is able to achieve their goal after playing a given game while taking into account the goal of the other player to \emph{compete or collaborate if needed}. Additionally, the paper introduces the notion of ``semi-competitiveness'' to characterize the players' behavior.

Syntax, semantics and a proof calculus for \dGLsc are given in this paper. Monotonicity is shown to hold for \dGLsc and the winning regions of \dGLsc coincide with those of zero-sum \dGL for complementary goals.
Furthermore, complementarization has been proven admissible, illustrating the semi-competitive behavior of the players and justifying that the semantics is defined suitably.

In the future, several extensions of \dGLsc would be interesting. First, players might behave suboptimally by not always cooperating when necessary. In this case, players would need to actively agree on cooperating. As zero-sum three player games can be reduced to non-zero sum two-player games \cite{vonNeumannMorgenstern+2004}, \dGLsc could also be extended to a three-player logic. This logic could then be used for safety proofs of more complex scenarios with more players like overtaking maneuvers.

\paragraph{Acknowledgements.}
This work was funded by an Alexander von Humboldt Professorship and by KiKIT, the Pilot Program for Core-Informatics at the KIT of the Helmholtz Association.

\bibliography{refs}{}
\bibliographystyle{splncs04}
\appendix
\section{Appendix}
\begin{definition}[\dGLsc Semantics]
    \label{semantics}
    The semantics of \dGLsc formulas is defined as the set of states in which they are true as follows:
    \begin{itemize}
        \item $\winr{e\geq \Tilde{e}} = \{ \omega \in \s \;|\; \omega \winr{e} \geq \omega\winr{\Tilde{e}}\}$ where $\omega \winr{e}$ is the evaluation of term $e$ in the state $\omega$.
        \item $\winr{\lnot P} = \s\setminus \winr{P} = \winr{P}^\C$
        \item $\winr{P \land Q} = \winr{P} \cap \winr{Q}$
        \item $\winr{\forall x P} = \{\omega \in \s \;|\; \omega_x^r \in \winr{P} \text{ for all } r \in \mathbb{R}\}$
        \item $\winr{\exists x P} = \{ \omega \in \s \;|\; \omega_x^r \in \winr{P} \text{ for some } r \in \mathbb{R}\}$
        \item $\winr{\angel{\alpha}{P}{Q}} = \winrA{\alpha}{\winr{P}}{\winr{Q}}$
        \item $\winr{\demon{\alpha}{P}{Q}} = \winrD{\alpha}{\winr{P}}{\winr{Q}} $
    \end{itemize}
    where $e, \Tilde{e}$ are terms, $P,Q$ are formulas, $x$ is a variable and $\alpha$ is a hybrid game.
    The function $\winrA{\alpha}{\cdot}{\cdot}$ is defined inductively as:
    \begin{itemize}
        \item $\winrA{x:=e}{X}{Y} = \{\omega \in \s \;|\; \omega_x^{\omega\winr{e}} \in X\}$
        \item $\winrA{x'=f(x) \& Q}{X}{Y} = \{\varphi(0) \;|\; \varphi(r) \in X \text{ for some } r \text{ with } \varphi \models x'=f(x) \land Q\}$
        \item $\winrA{?Q}{X}{Y} = \winr{Q} \cap X$
        \item $\winrA{\alpha \cup \beta}{X}{Y}=\winrA{\alpha}{X}{Y} \cup \winrA{\beta}{X}{Y}$
        \item $\winrA{\alpha;\beta}{X}{Y}=\winrA{\alpha}{\winrA{\beta}{X}{Y}}{\winrD{\beta}{X}{Y}}$
        \item $\winrA{\alpha^d}{X}{Y} = \winrD{\alpha}{Y}{X}$
        \item $\begin{aligned}[t]
            \winrA{\alpha^*}{X}{Y} &=\begin{aligned}[t]
                &\bigcap\{Z\subseteq\s\;|\; X \cup \winrA{\alpha}{Z}{Z^\C} \subseteq Z\} \\&\cup \bigcap\{Z\subseteq\s\;|\; (X\cap Y) \cup (\winrA{\alpha}{Z}{Z}\cap \winrD{\alpha}{Z}{Z})\subseteq Z\}
            \end{aligned}
        \end{aligned}$
    \end{itemize}
    The function $\winrD{\alpha}{\cdot}{\cdot}$ is defined as:
    \begin{itemize}
        \item $\winrD{x:=e}{X}{Y} = \{\omega \in \s \;|\; \omega_x^{\omega\winr{e}} \in Y\}$
        \item $\begin{aligned}[t]
            \winrD{x'=f(x) \& Q}{X}{Y} &= \begin{aligned}[t]
                &\{\varphi(0) \in \s\;|\; \varphi(r) \in Y \text{ for all } r \text{ with } \varphi \models x'=f(x) \land Q\}\\
                &\cup \begin{aligned}[t]
                    &\{\varphi(0) \in \s \;|\; \varphi(r) \in X \cap Y \\
                    &\text{ for some } r \text{ with } \varphi \models x'=f(x) \land Q\}
                \end{aligned}
            \end{aligned}
        \end{aligned}$
        \item $\winrD{?Q}{X}{Y} = \winr{Q}^\C \cup Y$
        \item $\begin{aligned}[t]
            \winrD{\alpha \cup \beta}{X}{Y} &= \begin{aligned}[t]
                &(\winrD{\alpha}{X}{Y} \cap \winrD{\beta}{X}{Y})\\
                &\cup (\winrD{\alpha}{X}{Y} \cap \winrA{\alpha}{X}{Y})\\
                &\cup (\winrD{\beta}{X}{Y} \cap \winrA{\beta}{X}{Y})
            \end{aligned} 
        \end{aligned}$
        \item $\winrD{\alpha;\beta}{X}{Y} = \winrD{\alpha}{\winrA{\beta}{X}{Y}}{\winrD{\beta}{X}{Y}}$
        \item $\winrD{\alpha^d}{X}{Y} = \winrA{\alpha}{Y}{X}$
        \item $\begin{aligned}[t]
            \winrD{\alpha^*}{X}{Y} &= \begin{aligned}[t]
                &\bigcup\{Z\subseteq\s \;|\; Z \subseteq Y \cap \winrD{\alpha}{Z^\C}{Z}\}\\& \cup \bigcap\{Z\subseteq\s\;|\; (X\cap Y) \cup (\winrA{\alpha}{Z}{Z}\cap \winrD{\alpha}{Z}{Z})\subseteq Z\}
            \end{aligned}   
          \end{aligned}$
    \end{itemize}
\end{definition}
\printProofs
\end{document}